\documentclass[12pt, a4paper, leqno]{amsart}

\usepackage{amsmath,xparse}

\NewDocumentCommand{\oldnorm}{sO{}m}{%
  {\IfBooleanTF{#1}
    {\oldnormaux{\left|}{\right|}{#3}}
    {\oldnormaux{#2|}{#2|}{#3}}}
}
\makeatletter
\newcommand{\oldnormaux}[3]{\mathpalette\oldnormaux@i{{#1}{#2}{#3}}}
\newcommand{\oldnormaux@i}[2]{\oldnormaux@ii#1#2}
\newcommand{\oldnormaux@ii}[4]{%
  \sbox\z@{$\m@th#1#2#4#3$}%
  \sbox\tw@{$\m@th\|$}%
  \mathopen{\hbox to\wd\tw@{\hss\vrule height \ht\z@ depth \dp\z@ width .3\wd\tw@\hss}}%
  #4
  \mathclose{\hbox to\wd\tw@{\hss\vrule height \ht\z@ depth \dp\z@ width .3\wd\tw@\hss}}%
}
\makeatother

\usepackage{esint} 

\usepackage{amsthm}
\usepackage{amsfonts}
\usepackage{amssymb}
\usepackage{amscd}
\usepackage{mathrsfs}
\usepackage{enumerate}
\usepackage{bm}
\usepackage{dsfont}

\usepackage[utf8]{inputenc}
\usepackage[english]{babel}
\usepackage{anysize}
\usepackage{color}

\usepackage{lipsum}

\newcommand\blfootnote[1]{%
  \begingroup
  \renewcommand\thefootnote{}\footnote{#1}%
  \addtocounter{footnote}{-1}%
  \endgroup
}

\newcommand{\R}{\mathbb{R}}
\newcommand{\Op}{\operatorname{Op}}
\newcommand{\Ad}{\operatorname{Ad}}

\newcommand{\dist}{\operatorname{dist}}

\newcommand{\vertiii}[1]{{\left\vert\kern-0.25ex\left\vert\kern-0.25ex\left\vert #1 
    \right\vert\kern-0.25ex\right\vert\kern-0.25ex\right\vert}}

\newtheorem{prop}{Proposition}
\newtheorem{lemma}{Lemma}

\newtheorem{definition}{Definition}
\newtheorem{corol}{Corollary}
\newtheorem{teor}{Theorem}
\theoremstyle{remark}

\theoremstyle{remark}
\newtheorem{remark}{Remark}

\title[Spectral stability  for renormalized KAM systems]{Spectral stability and semiclassical measures for renormalized KAM systems}
\author{Víctor Arnaiz}

\begin{document}

\begin{abstract}
An exact semiclassical version of the classical KAM theorem about small perturbations of vector fields on the torus is given. Moreover, a \textit{renormalization} theorem based on \textit{counterterms} for some semiclassical systems that are close to  being completely integrable is obtained. We apply these results to characterize the sets of semiclassical measures and quantum limits for sequences of $L^2$-eigenfunctions of these systems.
\end{abstract}

\blfootnote{The author has been supported by La Caixa, Severo Ochoa ICMAT, International Phd. Programme, 2014, and MTM2017-85934-C3-3-P (MINECO, Spain). 
}

\maketitle

\section{Introduction}

\subsection{Motivation} 

Let $(M,g)$ be a compact and boundaryless Riemannian manifold, we consider the semiclassical Schrödinger equation
\begin{equation}
\label{e:eigenvalue_problem}
\widehat{H}_\hbar \, \Psi_{\hbar} = \lambda_{\hbar} \, \Psi_{\hbar}, \quad \Vert \Psi_{\hbar} \Vert_{L^2(M)}=1,
\end{equation}
where $\hbar \in (0,1]$ is a small parameter, given by a self\-adjoint opera\-tor $\widehat{H}_\hbar = \Op_\hbar(H)$ on $L^2(M)$ obtained as the \textit{semiclassical Weyl quantization} (see for instance \cite{Dim_Sjo99}, \cite{Zw12}  among many references) of a Hamiltonian function $H \in \mathcal{C}^\infty(T^*M;\R)$ defined on the cotangent bundle $T^*M$. A fundamental example to keep in mind is the semiclassical Schrödinger operator\footnote{We will not deal with this particular operator, but it is important to motivate our problem and connect it with several related works.}
\begin{equation}
\label{e:quantum_hamiltonian}
\widehat{H}_\hbar = -\hbar^2 \Delta_g + V(x),
\end{equation}
where $\Delta_g$ denotes the Laplace-Beltrami operator and we assume that the potential $V $ is smooth and real valued. This operator is the Weyl quantization of the classical Hamiltonian
\begin{equation}
\label{e:classic_hamiltonian}
H(x,\xi) := \Vert \xi \Vert_g^2 + V(x), \quad (x,\xi) \in T^*M,
\end{equation}
obtained as the sum of the kinetic and the potential energies.

We assume that the spectrum of $\widehat{H}_\hbar$ is pure-point and unbounded, meaning that there exist an orthonormal basis of $L^2(M)$ consisting of eigenfunctions for $\widehat{H}_\hbar$, and a sequence of eigenvalues $(\lambda_{\hbar,n})$ satisfying 
\begin{equation}
\label{e:property_of_spectrum}
\displaystyle \lim_{n \to + \infty} \lambda_{\hbar,n} = + \infty, \quad  \text{for all } \hbar \in (0,1].
\end{equation}
For example, if $\widehat{H}_\hbar$ is given by \eqref{e:quantum_hamiltonian} then its spectrum is indeed discrete, and given by a unique sequence of eigenvalues $(\lambda_{\hbar,n}) \subset \R$ satisfying \eqref{e:property_of_spectrum}. 

Therefore, for a given sequence $(n_\hbar) \subset \mathbb{N}$, and a given $E \in H(T^*M) \subset \R$, we can choose a decreasing-to-zero sequence of parameters $(\hbar) \subset (0,1]$ so that 
$$
\lambda_\hbar := \lambda_{n_\hbar,\hbar} \to E, \quad \text{as } \hbar \to 0^+.
$$
Modulo adding a constant $E_0$ to $H$, we can assume that $E = 1$.

We aim at understanding the accumulation points (in the weak-$\star$ topology for Radon measures) of those sequences of densities $\vert \Psi_{\hbar}(x) \vert^2dx$ as $\lambda_\hbar \to 1$. These limits are pro\-bability measures on $M$ and are usually referred to as \textit{quantum limits}. We will denote by $\mathcal{N}(\widehat{H}_\hbar)$ the set of quantum limits of $\widehat{H}_\hbar$.

The problem of characterizing the set $\mathcal{N}(\widehat{H}_\hbar)$ is in general widely open, but it is well known that the elements of  $\mathcal{N}(\widehat{H}_\hbar)$ depend strongly on the classical dynamics generated by the Hamiltonian $H$. Recall that $H$ generates a dynamical system on $T^*M$ via the Hamilton equations 
$$
\dot{x}(t) = \partial_\xi H, \quad \dot{\xi}(t) = - \partial_x H, \quad ( x(0),\xi(0)) = (x_0, \xi_0) \in T^*M.
$$
We will denote by $\phi_t^H$ the Hamiltonian flow generated by $H$, that is,
$$
\phi_t^H(x_0,\xi_0) = (x(t),\xi(t)), \quad t \in \R.
$$
Note that, for the free Schrödinger operator $\widehat{H}_\hbar = -\hbar^2 \Delta_g$, the associated classical Hamiltonian flow $\phi_t^H$ is nothing but the geodesic flow on $T^*M$. 

Mostly three cases  have been studied so far: the case when $\phi_t^H$ is ergodic with respect to the Liouville measure, the case when $H$ generates a  completely integrable system, and the case when $\phi_t^H$ lies in some mixed or KAM (Kolmogorov-Arnold-Moser) regime.

In this work we focus on systems that are close to completely integrable ones, for which KAM techniques apply. For these kind of systems, the persistence of invariant tori by the dynamics of the classic flow implies a weak dispersive behavior of the Schrödinger flow around these tori. As a consequence, one expects the existence of sequences of eigenfunctions for $\widehat{H}_\hbar$ concentrating on these tori. This concentration takes place on the phase space $T^*M$, meaning that the \textit{semiclassical measure} of the sequence (see Section \ref{s:results} below, \cite{Ger90} among many references)  has positive mass on these tori.

In the completely integrable setting, if $\widehat{H}_\hbar = - \hbar^2 \Delta_g$ and  $ M = \mathbb{S}^d$, the sphere endowed with its canonical metric, Jakobson and Zelditch \cite{Jak_Zel99} proved that:
\begin{equation}\label{e:maximal}
\mathcal{N}(\widehat{H}_\hbar)=\overline{ \text{Conv}\{\delta_\gamma\,:\,\gamma\text{ is a closed  geodesic orbit in }\mathbb{S}^d\}}.
\end{equation} 
Above, $\delta_\gamma$ stands for the uniform probability measure on the closed curve $\gamma$\footnote{$\gamma = \pi (\sigma)$ is the projection onto $\mathbb{S}^d$ of a periodic geodesic $\sigma \subset T^*\mathbb{S}^d$, i.e. $\sigma$ is a minimal invariant torus by $\phi_t^H$ of dimension one.}. Property \eqref{e:maximal} also holds in manifolds of constant positive curvature \cite{AM:10} or compact-rank-one symmetric spaces \cite{MaciaZoll}.  A natural question in this setting is that of understanding whether or not the same holds on a Zoll manifold (that is, a manifold all whose geodesics are closed \cite{BesseZoll},  which is still a completely integrable system). Macià and Rivière \cite{MaciaRiviere16}, \cite{MaciaRiviere17} have shown the existence of Zoll surfaces such that \eqref{e:maximal} fails. Precisely, an open set of geodesics is excluded to be the support of any quantum limit; that is, the delta measure $\delta_\gamma$ can not be a quantum limit for any geodesic $\gamma$ in this open set. Similar techniques as those of \cite{MaciaRiviere16}, \cite{MaciaRiviere17} have been used in the study of spectral asymptotics for small perturbations of harmonics oscillators, both in the selfadjoint case \cite{Ar_Mac18} and the non-selfadjoint case \cite{Ar_Riv18}.

On the flat torus $\mathbb{T}^d:=\mathbb{R}^d/2\pi\mathbb{Z}^d$, the behavior of quantum limits is very different. Bourgain proved that $\mathcal{N}(\widehat{H}_\hbar)\subset L^1(\mathbb{T}^d)$; and in particular that quantum limits cannot concentrate on closed curves, as was the case on the sphere (this result was reported in \cite{Jak97}). In that same reference, Jakobson proved that for $d=2$ the density of any quantum limit is a trigonometric polynomial, whose frequencies satisfy a certain Pell equation. In higher dimensions, one can only prove certain regularity properties of the densities, involving decay of its Fourier coefficients.  These and related results were proven using only  the dynamical properties of the geodesic flow by Macià and Anantharaman \cite{An_Mac14}, \cite{MaciaTorus}, \cite{MaciaDispersion}; it is also possible to obtain more precise results on the regularity of the densities \cite{AJM}. This strategy of proof  can be extended to more general completely integrable Hamiltonian flows \cite{An_Fer_Mac15}, and also allows to deal with domains in the Euclidean space as disks \cite{ALM:16}, \cite{ALMCras}. 

The KAM regime has turned out to be more elusive so far.  Most of the works dealing with this case are based on the construction of \textit{quasimodes}, or approximate eigenfunctions, studying the asymptotic properties of oscillation and concentration of these quasimodes around the classical invariant tori,  but do not conclude complete results for the quantum limits associated with the true eigenfunctions of the system. The foundations of this study of quasimodes for KAM systems can be found in Lazutkin \cite{Laz93}. Construction of quasimodes with exponentially small error terms is given by Popov \cite{Pop00I}, \cite{Pop00II}. In a very recent work, Gomes \cite{Gomes18} applies this result to discard \textit{quantum ergodicity} for  semiclassical KAM systems on compact Riemannian manifolds. Moreover, in two dimensions, Gomes and Hassell \cite{Gom_Hass18} improve the previous result to show that there exist sequences of eigenfunctions with semiclassical measure having positive mass on the classical invariant tori. 

The present work addresses the problem of characterizing the sets of quantum limits and semiclassical measures for some perturbed integrable systems, coming from KAM theory, that can be \textit{renormalized}; that is, they can be conjugated  to the unperturbed system after  adding integrable counterterms to the principal part of the perturbed Hamiltonian. Our techniques differ from those of \cite{Gomes18}, \cite{Gom_Hass18}. Precisely, we do not work from any quasimode construction. Alternatively, we show convergence of the quantum normal form via an iterative algorithm which is shown to converge in the presence of counterterms.

\subsection{Quantum limits and semiclassical measures for KAM families of vector fields on the torus}\label{s:results} From now on we fix $M = \mathbb{T}^d$, the flat torus endowed with the flat metric. Our first and particularly simple example of KAM system will be the one generated by the Schrödinger operator 
$$
\widehat{H}_\hbar = \widehat{P}_{\omega,\hbar} := \omega \cdot \hbar D_x + v(x;\omega) \cdot \hbar D_x - \frac{i\hbar}{2} \operatorname{Div} v(x;\omega),
$$
where $\omega \in \R^d$, $v \in \mathcal{C}^\infty(\mathbb{T}^d \times \R^d;\R^d)$ is a  vector field depending on the parameter $\omega$, and we use the notation
$$
D_x = (D_{x_1}, \ldots , D_{x_d}), \quad D_{x_j} := -i\partial_{x_j}.
$$
This operator generates the transport along the vector field $X_v(\omega) := \omega + v(\cdot;\omega)$, meaning that the solution to the Schrödinger equation
$$
\big( i\hbar \hspace*{0.05cm} \partial_t + \widehat{P}_{\omega,\hbar} \big) u_\hbar(t,x) = 0; \quad u_\hbar(0,x) = u^0_\hbar(x) \in L^2(\mathbb{T}^d)
$$
is given by
$$
u_\hbar(t,x) = u^0_\hbar\big( \phi_t^{X_v(\omega)}(x) \big) \sqrt{ \vert \det d\phi_t^{X_v(\omega)}(x) \vert},
$$
where $\phi_t^{X_v(\omega)}$ is the flow  on $\mathbb{T}^d$ generated by the vector field $X_v(\omega)$, and the operator $\widehat{P}_{\omega,\hbar}$ is selfadjoint thanks to the component $-i\hbar \operatorname{Div} v/2$. Note that the unperturbed operator
\begin{equation}
\label{e:non_perturbed_oprator}
\widehat{L}_{\omega,\hbar} := \omega \cdot \hbar D_x
\end{equation} on $L^2(\mathbb{T}^d)$ is not elliptic and hence its point-spectrum, given by 
$$
\operatorname{Sp}^p_{L^2(\mathbb{T}^d)} \big( \widehat{L}_{\omega, \hbar} \big) = \{ \hbar \, \omega \cdot k \, : \, k \in \mathbb{Z}^d \},
$$
is highly unstable under perturbations, in the sense that it could be transformed into continuous spectrum by the perturbation. However, we will use classical KAM theory to show that under certain conditions on the perturbation $v$, the spectrum of $\widehat{P}_{\omega,\hbar}$ is stable for a Cantor set of frequencies $\omega$, modulo renormalization of the vector $\omega$.  As was shown by Wenyi and Chi in \cite{Wen00}, this KAM stability is equivalent to the hypoellipticity of the operator $\widehat{P}_{\omega,\hbar}$.


On the other hand, the operator $\widehat{P}_{\omega,\hbar} = \Op_\hbar(P_\omega)$ is the semiclassical Weyl quantization of the linear Hamiltonian
$$
P_\omega(x,\xi) = \mathcal{L}_\omega(\xi) + v(x;\omega) \cdot \xi,
$$
where
$$
\mathcal{L}_\omega(\xi) := \omega \cdot \xi.
$$
In \cite{Mos67}, Moser introduced a new approach to the study of quasiperiodic motions by considering the frequencies of the Kronecker tori as independent parameters. We refer to the work of Pöschel \cite{Pos09} for a brief introduction to the subject. If $\Omega \subset \R^d$ is a compact Cantor set of frequencies satisfying some Diophantine condition (see condition \eqref{e:diophantine_condition2} below) and the perturbation $v$ is suffi\-ciently small in some suitable norm, then there exists a close-to-the-identity change of frequencies 
\begin{align*}
\varphi : \Omega \to \R^d
\end{align*}
so that the related set of Hamiltonians $P_{\varphi(\omega)}$ can be canonically conjugated (frequency by frequency) into the constant linear Hamiltonian on $T^*\mathbb{T}^d$ with frequency $\omega$. More precisely, for every $\omega \in \Omega$ there exists a canonical transformation $\Theta_\omega : T^*\mathbb{T}^d \to T^*\mathbb{T}^d$ so that
$$
\Theta_\omega^* P_{\varphi(\omega)}(x,\xi) = \mathcal{L}_\omega(\xi).
$$
In particular, the Hamiltonian $P_{\varphi(\omega)}$ is completely integrable for every $\omega \in \Omega$.

We focus on the study of the high-energy structure of the eigenfunctions of $\widehat{P}_{\omega,\hbar}$. Precisely, we will study the set of \textit{quantum limits} of the system, that is, the weak-$\star$ accumulation points of sequences of $L^2$-densities of eigenfunctions. 

Furthermore, it is very convenient to extend our analysis to the phase-space, studying not only the asymptotic distribution of  $L^2$-densities of the sequence $(\Psi_\hbar)$ on $\mathbb{T}^d$, but the related sequence of Wigner distributions $(W_{\Psi_\hbar}^\hbar)$ on  $T^*\mathbb{T}^d$.  

We recall that the Wigner distribution $W^\hbar_{\psi}$ of a function $\psi \in L^2(\mathbb{T}^d)$ is defined by the map
\begin{equation}
\label{e:wigner_distribution}
W^\hbar_{\psi} \, : \, \mathcal{C}_c^\infty(T^*\mathbb{T}^d) \ni a \longmapsto \big \langle \psi, \Op_\hbar(a) \psi \big \rangle_{L^2(\mathbb{T}^d)}.
\end{equation}
Since $\Op_\hbar(a)$ is bounded on $L^2(\mathbb{T}^d)$ uniformly in $\hbar \in (0, 1]$ in terms of the $L^\infty$-norms of a finite number of derivatives of $a$, for any sequence $(\psi_\hbar)_\hbar \subset L^2(\mathbb{T}^d)$ with $\Vert \psi_\hbar \Vert_{L^2} = 1$, there exist a subsequence $(W^\hbar_{\psi_{\hbar}})$ of Wigner distributions, and a distribution $\mu \in \mathcal{D}'(T^*\mathbb{T}^d)$ such that
$$
\lim_{\hbar \to 0} W^\hbar_{\psi_{\hbar}}(a) = \mu(a), \quad \forall a \in \mathcal{C}_c^\infty(T^*\mathbb{T}^d).
$$
Furthermore, the distribution $\mu$ is certainly a positive Radon measure on $T^*\mathbb{T}^d$ \cite{Ger90}. The measure $\mu$ is called the semiclassical measure associated to the (sub)sequence $(\psi_{\hbar})$. 

If $\mu$ is the semiclassical measure associated with a sequence of eigenfunctions $(\Psi_{\hbar})$ with $\lambda_{\hbar} \to 1$, then $\mu$ is in fact a positive Radon measure on the level-set $P_\omega^{-1}(1) \subset T^*\mathbb{T}^d$.  If moreover the measure $\mu$ turns out to be a probability measure, then its projection onto the position space 
$$
\nu(x) = \int_{P_\omega^{-1}(1)} \mu(x,d\xi)
$$
is the quantum limit of the sequence. We emphasize that, since $P_\omega^{-1}(1)$ is in general not compact, there can exist some sequences of eigenfunctions with the zero measure as semiclassical measure. We will denote by $\mathcal{M}(\widehat{P}_{\omega,\hbar})$ the set of semiclassical measures associated to sequences of eigenfunctions for $\widehat{P}_{\omega,\hbar}$ with $\lambda_\hbar \to 1$.

From now on, we consider $\widehat{P}_{\omega,\hbar}$ with frequencies $\omega$ lying in a small neighborhood of a compact Cantor set of Diophantine vectors $\Omega \subset \R^d$ satisfying:
\begin{equation}
\label{e:diophantine_condition2}
\vert k \cdot \omega \vert \geq \frac{ \varsigma}{\vert k \vert^{\gamma-1}}, \quad k \in \mathbb{Z}^d \setminus \{ 0 \},
\end{equation}
for some constants $\varsigma > 0$ and $\gamma > d$.  For any $\rho > 0$, let $\Omega_\rho$ be the complex neighborhood of $\Omega$ given by
$$
\Omega_\rho := \{ z \in \mathbb{C}^d \, : \, \dist(z,\Omega) < \rho \},
$$
and, given $s >0$, we consider also the complex neighborhood of the $d$-torus
$$
D_s := \{ z \in \mathbb{T}^d + i \R^d \, : \, \vert \Im z \vert < s \}.
$$

We assume that the perturbation $V(x,\xi;\omega) :=v(x;\omega) \cdot \xi$ belong to the following family of linear symbols on the cotangent bundle $T^*\mathbb{T}^d$:

\begin{definition}
A function $V \in \mathcal{C}^\omega(T^*\mathbb{T}^d \times \Omega_\rho)$ is in the space of linear symbols $\mathscr{L}_{s,\rho}$ if 
\begin{equation}
V(x, \xi; \omega) = \xi \cdot v(x;\omega) = \sum_{k \in \mathbb{Z}^d}  \xi \cdot \widehat{v}(k;\omega)  e_k(x),
\end{equation}
for some analytic vector field $v \in \mathcal{C}^\omega(D_s \times \Omega_\rho; \mathbb{C}^d)$, where $\widehat{v}(k;\omega) \in \mathbb{C}^d$ is the $k^{th}$-Fourier coefficient of $v$:
$$
\widehat{v}(k;\omega) := \big \langle v(\cdot;\omega) , e_{k} \big \rangle_{L^2(\mathbb{T}^d)}, \quad e_k(x) := \frac{e^{ik \cdot x}}{(2\pi)^{d/2}}, \quad k \in \mathbb{Z}^d,
$$
and
\begin{equation}
\vert V \vert_{s,\rho} := \sup_{\omega \in \Omega_\rho}  \sum_{k \in \mathbb{Z}^d} \vert \widehat{v}(k;\omega) \vert e^{\vert k \vert s} < \infty.
\end{equation}
The space $\big( \mathscr{L}_{s,\rho}, \vert \cdot \vert_{s,\rho} \big)$ is a Banach space. We denote by $\mathscr{L}_s \subset \mathscr{L}_{s, \rho}$ the subspace of symbols that do not depend on $\omega \in \Omega_\rho$, and by $\vert \cdot \vert_s$ its norm in this space.
\end{definition}

Our first result reads:

\begin{teor}
\label{t:main_result_1}
Let $s,\rho > 0$ and $V \in \mathscr{L}_{s,\rho}$ be real valued and assume
\begin{equation}
\label{e:small_condition}
\vert V \vert_{s,\rho}  \leq \varepsilon,
\end{equation}
where $\varepsilon$ is a small positive constant depending only on $s$, $\rho$, $\gamma$ and $\varsigma$. Then there exists a real change of frequencies $\varphi : \Omega \to \Omega_{\rho}$ such that the point-spectrum of $\widehat{P}_{\varphi(\omega),\hbar}$ is
$$
\operatorname{Sp}^p_{L^2(\mathbb{T}^d)}\big( \widehat{P}_{\varphi(\omega),\hbar} \big) = \{ \hbar \, \omega \cdot k \, : \, k \in \mathbb{Z}^d \},
$$
and, for every $\omega \in \Omega$, there exists a diffeomorphism $\theta_\omega : \mathbb{T}^d \to \mathbb{T}^d$ of the torus homotopic to the identity so that, denoting by
$$
\Theta_\omega(x,\xi) = \big( \theta_\omega(x), [(\partial_x \theta_\omega(x))^T]^{-1} \xi \big)
$$
the symplectic lift of $\theta_\omega$ into $T^*\mathbb{T}^d$, 
$$
\mathcal{M}\big( \widehat{P}_{\varphi(\omega),\hbar} \big) =  \bigcup_{\xi \in \mathcal{L}_\omega^{-1}(1)} \big \{ (\Theta_\omega)_* \mathfrak{h}_{\mathbb{T}^d \times \{ \xi \}} \big \} \cup \{0 \},
$$
where $\mathfrak{h}_{\mathbb{T}^d \times \{ \xi \}}$ denotes the Haar measure on the invariant torus $\mathbb{T}^d \times \{ \xi \}$ and $(\Theta_\omega)_*$ stands for the pushforward of $\Theta_\omega$;
and
$$
\mathcal{N}(\widehat{P}_{\varphi(\omega),\hbar}) = \left \{ \frac{1}{(2\pi)^d}(\theta_\omega)_* dx \right \}.
$$
Moreover,  
$$
\sup_{\omega \in \Omega} \vert \varphi(\omega) - \omega \vert \leq C_1 \vert V \vert_{s,\rho}, \quad 
\sup_{x \in \mathbb{T}^d} \vert \theta_\omega(x) - x \vert \leq C_2\vert V \vert_{s,\rho},
$$
where $C_1$ and $C_2$ are positive constants depending only on $s$, $\rho$, $\gamma$ and $\varsigma$. 
\end{teor}

\begin{remark}
The assumption $V \in \mathscr{L}_{s,\rho}$ allows us to use an analytic version of the classical KAM theorem (Theorem \ref{poschel} below) about perturbations of constant vector fields on the torus \cite{Mos67}, \cite{Pos11}. This theorem remains valid for less regular symbols $V$, see for instance \cite{Mos67}, \cite{Zehnder75I}, \cite{Zehnder75II}, as well as for more general Diophantine conditions than \eqref{e:diophantine_condition2}, see  \cite{Pos09}, \cite{Pos11}. We prefer not to state our result with the greatest possible generality for the sake of clarity.
\end{remark}

\subsection{Renormalization of semiclassical KAM operators} If the perturbation $V$ does not depend on the vector of frequencies and we consider an isolated vector $\omega \in \Omega$, then Theorem \ref{t:main_result_1} provides the following direct corollary:

\begin{corol}
Let $\omega \in \Omega$ and let $V \in \mathscr{L}_s$ such that
$$
\vert V \vert_s \leq \varepsilon,
$$
where $\varepsilon$ is a small positive constant depending only on $s$, $\gamma$ and $\varsigma$. Then there exist a real vector  $\lambda = \lambda(V) \in \R^d$ such that the point-spectrum of $\widehat{P}_{\omega + \lambda,\hbar}$ is
$$
\operatorname{Sp}^p_{L^2(\mathbb{T}^d)}\big( \widehat{P}_{\omega + \lambda,\hbar} \big) = \{ \hbar \, \omega \cdot k \, : \, k \in \mathbb{Z}^d \},
$$
and a diffeomorphism $\theta : \mathbb{T}^d \to \mathbb{T}^d$ of the torus homotopic to the identity so that 
$$
\mathcal{M}\big( \widehat{P}_{\omega + \lambda,\hbar} \big) =  \bigcup_{\xi \in \mathcal{L}_\omega^{-1}(1)} \big \{ \Theta_* \mathfrak{h}_{\mathbb{T}^d \times \{ \xi \}} \big \} \cup \{0 \}, \quad \mathcal{N}(\widehat{P}_{\omega + \lambda,\hbar}) = \left \{ \frac{1}{(2\pi)^d}\theta_* dx \right \},
$$
where $\Theta$ is the symplectic lift of $\theta$ into $T^*\mathbb{T}^d$.
\end{corol}
The vector $\lambda \in { \R^d}$ can be understood as a \textit{counterterm} that renormalizes the perturbed operator $\widehat{P}_{\omega,\hbar}$ to make it completely integrable and unitarily equivalent to $\widehat{L}_{\omega,\hbar}$.

In the classical framework, the renor\-ma\-lization problem \cite{Gall82}, \cite{Gen96} asks if, given a small analytic perturbation $V$ of the linear Hamiltonian $\mathcal{L}_\omega$, with $V = V(x,\xi;\varepsilon)$ defined on $\mathbb{T}^d \times \R^d \times [0,\varepsilon_0]$ for some $\varepsilon_0>0$ sufficiently small, there exists a counterterm $R = R(\xi; \varepsilon)$ on $\R^d \times [0,\varepsilon_0]$, such that the renormalized Hamiltonian
$$
Q(x,\xi;\varepsilon) = \mathcal{L}_\omega(\xi) + V(x,\xi;\varepsilon) - R(\xi;\varepsilon)
$$
is integrable and canonically conjugate to the unperturbed Hamiltonian. This was conjectured by Gallavotti in \cite{Gall82} and first proven by Eliasson in \cite{Elia89}. This result can be regarded as a control theory theorem. Despite the fact that small perturbations of $\mathcal{L}_\omega$ could gene\-rate even ergodic behavior (see Katok \cite{Kat73}), this shows that modifying in a suitable way the completely integrable part of the Hamiltonian, the system remains stable. Renormalization techniques have been studied by several authors in the context of quantum field theory, as well as its connection with KAM theory \cite{Chan98, Feld92, Gall82, Gall95, Koch99, Khan88}. 

Our goal is to prove a semiclassical version of the renormalization problem. We consider again the semiclassical Weyl quantization of $\mathcal{L}_\omega$:
\begin{equation}
\widehat{L}_{\omega, \hbar} = \Op_\hbar(\mathcal{L}_\omega)=  \omega \cdot \hbar D_x.
\end{equation}
Let $(\varepsilon_\hbar)_\hbar$ be a semiclassical scaling such that
\begin{equation}
\label{e:semiclassical_scaling}
 \varepsilon_\hbar \leq \hbar,
\end{equation} 
and let $V \in \mathcal{C}(T^*\mathbb{T}^d;\R)$ be a bounded real function. The precise assumptions on the regularity of $V$ will be stated below. Our aim is to construct an integrable counterterm $R_\hbar = R_\hbar(V) \in \mathcal{C}(\R^d)$, that only depends on the action variable $\xi$ and is uniformly bounded in $\hbar \in (0, 1]$, so that the quantum Hamiltonian
\begin{equation}
\widehat{Q}_\hbar := \widehat{L}_{\omega,\hbar} + \varepsilon_\hbar \Op_\hbar(V - R_\hbar)
\end{equation}
is unitarily equivalent to the unperturbed operator $\widehat{L}_{\omega,\hbar}$. This will show that the spectrum of the operator $\widehat{L}_{\omega,\hbar} + \varepsilon_\hbar \Op_\hbar(V)$ can be stabilized by adding the counterterm $\varepsilon_\hbar \Op_\hbar(R_\hbar)$ to the system. Moreover, we will show that the sets of quantum limits and semi\-classical measures of sequences of eigenfunctions for the operator $\widehat{Q}_\hbar$ coincide with those of the unperturbed operator $\widehat{L}_{\omega,\hbar}$.

In a related work, Graffi and Paul \cite{Pau12} showed that the perturbed operator 
$$
\widehat{P}_\hbar = \widehat{L}_{\omega,\hbar} +  \Op_\hbar(V_\omega)
$$
can be conjugated to a convergent quantum normal form for a specific class of bounded analytic perturbations of the form 
\begin{equation}
\label{Russmann_condition_perturbation}
V_\omega(x,\xi) = V(x, \omega \cdot \xi), \quad (x,\xi) \in T^*\mathbb{T}^d,
\end{equation}
(see Gallavotti \cite{Gall82} for a discussion of this condition). As a consequence, it is most likely that the set of semiclassical measures is stable under perturbations of this type, without necessity of renormalization.  The main difference in our approach is the substitution of the particular dependence on $\omega \cdot \xi$ of $V$, which is stable under the conjugacies employed by Graffi and Paul to construct the normal form, by the addition of the renormalization function $R_\hbar$.

We emphasize that, compared to \cite{Eca98}, \cite{Elia89} and \cite{Gen96}, our work is not based on the study of the convergence of Lindstedt series. Alternatively, we will use an algorithm similar to that of Govin et al. \cite{Gov98} to construct a  normal form, obtaining the counterterm $R_\hbar$ step by step. We expect that condition \eqref{e:semiclassical_scaling} is not sharp. One should be able to deal with perturbations of order $O(1)$. The main difficulty arises when managing the loss of analyticity in the variable $\xi$.

We will consider semiclassical perturbations  $\Op_\hbar(V)$ whose symbol $V$ enjoys some regularity  properties. We first consider the following spaces of analytic functions:

\begin{definition}
\label{Banach_spaces_of_symbols}
Let $s >0$, we define the Banach space $\mathcal{A}_s(\R^d)$ of analytic functions $f \in \mathcal{C}^\omega(\R^d;\mathbb{R})$ such that
$$
\Vert f \Vert_{\mathcal{A}_s(\R^d)} := \frac{1}{(2\pi)^{d/2}}\int_{\R^d} \vert \widehat{f}(\eta) \vert \, e^{\vert \eta \vert s} \, d\eta < \infty,
$$
where $\widehat{f}$ denotes the Fourier transform of $f$. Let $\rho > 0$, we also define the space $\mathcal{A}_{s,\rho}(T^* \mathbb{T}^d)$ of analytic functions $g \in \mathcal{C}^\omega(T^*\mathbb{T}^d;\mathbb{R})$ such that
$$
\Vert g \Vert_{s,\rho} := \frac{1}{(2\pi)^{d/2}} \sum_{k\in \mathbb{Z}^d} \Vert \widehat{g}(k,\cdot) \Vert_{ \mathcal{A}_s(\R^d)} \, e^{\vert k \vert \rho} < \infty,
$$
where
$$
\widehat{g}(k , \xi) := \frac{1}{(2\pi)^{d/2}} \int_{\mathbb{T}^d} g(x,\xi) e^{-i x\cdot k} \, dx, \quad k \in \mathbb{Z}^d.
$$
Finally, we define the space $\mathcal{A}_\rho(\mathbb{T}^d)$ of functions $v \in \mathcal{C}^\omega(\mathbb{T}^d;\R)$ such that
$$
\Vert v \Vert_{\mathcal{A}_\rho(\mathbb{T}^d)} := \frac{1}{(2\pi)^{d/2}} \sum_{k \in \mathbb{Z}^d} \vert \widehat{v}(k) \vert e^{\vert k \vert \rho} < \infty.
$$
\end{definition}

\noindent By the Calderón-Vaillancourt Theorem (see Lemma \ref{l:analytic_calderon_vaillancourt_torus} below), the semiclassical Weyl quantization $\Op_h(a)$ of a symbol $a \in \mathcal{A}_{s,\rho}(T^* \mathbb{T}^d)$ satisfies
$$
\Vert \Op_\hbar(a) \Vert_{\mathcal{L}(L^2)} \leq C_{d,\rho} \Vert a \Vert_{s,\rho}, \quad \forall \hbar \in (0,1].
$$

We next proceed to state our second result:

\begin{teor}
\label{t:main_result_2} 
Let $\omega \in \R^d$ be a Diophantine vector satisfying \textnormal{(\ref{e:diophantine_condition2})}, and let $V$ be a real valued function that belongs to $\mathcal{A}_{s,\rho}(T^*\mathbb{T}^d)$  for some fixed $s,\rho>0$. Assume that $\varepsilon_\hbar = \hbar$ and
\begin{equation}
\label{e:small_condition2}
\Vert V \Vert_{s,\rho} \leq \varepsilon,
\end{equation}
where $\varepsilon > 0$ is a small constant that depends only on $s$, $\rho$, $\gamma$ and $\varsigma$. Then, there exists a sequence of integrable counterterms\footnote{That is, $R_\hbar$ is a function only of the action variable $\xi \in \R^d$.}  $R_\hbar = R_\hbar(V) \in \mathcal{A}_{s}(\R^d)$ with $\Vert R_\hbar \Vert_{\mathcal{A}_s(\R^d)} \lesssim \Vert V \Vert_{s,\rho}$, uniformly in $\hbar \in (0,1]$, and
\begin{equation}
\label{e:stability_spectrum}
\operatorname{Sp}^p_{L^2(\mathbb{T}^d)}\big( \widehat{Q}_\hbar \big) = \operatorname{Sp}^p_{L^2(\mathbb{T}^d)}\big( \widehat{L}_{\omega,\hbar} \big)=  \{ \hbar \, \omega \cdot k \, : \, k \in \mathbb{Z}^d \}.
\end{equation}
Moreover, denoting by $\mathcal{M}\big(\widehat{Q}_\hbar \big)$ the set semiclassical measures of sequences of norma\-lized eigenfunctions of the Hamiltonian $\widehat{Q}_\hbar$ with eigenva\-lues verifying $\lambda_\hbar \to 1$ as $\hbar \to 0$, 
\begin{equation}
\label{e:semiclassical_measures}
 \mathcal{M}\big(\widehat{Q}_\hbar \big) = \mathcal{M}\big( \widehat{L}_{\omega,\hbar} \big) = \bigcup_{\xi \in \mathcal{L}_\omega^{-1}(1)} \big \{ \mathfrak{h}_{\mathbb{T}^d \times \{ \xi \}} \big \} \cup \{0 \},
\end{equation}
and the set of quantum limits of $\widehat{Q}_\hbar$ is precisely
\begin{equation}
\label{e:quantum_limits}
\mathcal{N}(\widehat{Q}_\hbar) = \left \{ \frac{1}{(2\pi)^d}dx \right \}.
\end{equation}
\end{teor}

In the case $\varepsilon_\hbar \ll \hbar$, condition \eqref{e:small_condition2} can be removed. Moreover, in this case we can relax the regularity hypothesis on $V$ in the $\xi$ variable.  Essentially, we just need that $\xi \mapsto V(x,\xi)$ is bounded together with all its derivatives to ensure that $\Op_\hbar(V)$ is bounded on $L^2(T^*\mathbb{T}^d)$. In \cite{Sjo95}, \cite{Sjo94}, Sjöstrand introduced a Wiener algebra of $L^2$-pseudodifferential operators with symbols containing the class $S^0(\R^{2d})$ (the class of symbols $a \in \mathcal{C}^\infty(\R^{2d})$ that are bounded together with all its derivatives) that is stable under composition. From these works, we define the following spaces:
\begin{definition}
Let $e_1, \ldots , e_d$ be a basis of $\R^d$, we say that $\Gamma =  \bigoplus_1^d \mathbb{Z} e_j$ is a lattice. Let $\chi_0 \in \mathscr{S}(\R^{2d})$ have the property that $1 = \sum_{j \in \Gamma} \chi_j$, where $\chi_j(\xi) = \chi_0(x-j)$. We say that $f \in S_W(\R^d)$ if
$$
U(\eta) := \sup_{j \in \Gamma} \vert \widehat{ \chi_j f}(\eta) \vert \in L^1(\R^d).
$$
The space $S_W$ is a Banach space with the norm
\begin{equation}
\label{e:large_constant}
\Vert f \Vert_W := C_d \Vert \sup_{j \in \Gamma} \vert \widehat{ \chi_j f} \vert \Vert_{L^1(\R^d)},
\end{equation}
where $C_d$ is a fixed and large constant depending only on $d$. Let $\rho > 0$, we say that $g \in \mathcal{C}^\infty(\R^{2d};\R)$ belong to $\mathcal{A}_{W,\rho}(\R^{2d})$ if 
$$
\Vert g \Vert_{W,\rho} := \frac{1}{(2\pi)^{d/2}} \sum_{k \in \mathbb{Z}^d} \Vert \widehat{g}(k, \cdot) \Vert_W \, e^{\vert k \vert \rho} < \infty.
$$
\end{definition}
\begin{remark}
The definition of $S_W(\R^d)$ does not depend on the choice of $\Gamma$, $\chi_0$, see \cite[Lemma 1.1]{Sjo95}.
\end{remark}

\begin{remark}
Observe that, for every $s,\rho > 0$, $\mathcal{A}_{s,\rho}(T^*\mathbb{T}^d) \subset \mathcal{A}_{W,\rho}(T^*\mathbb{T}^d)$ and $\mathcal{A}_s(\R^d) \subset S_W(\R^d)$. Notice also that $\mathcal{A}_{W,\rho}(T^*\mathbb{T}^d)$ is contained in the space of bounded continuous functions on $T^*\mathbb{T}^d$.
\end{remark}

\begin{teor}
\label{t:main_result_3}
Let $\omega \in \R^d$ be a strongly non-resonant frequency satisfying \textnormal{(\ref{e:diophantine_condition2})}, and let $V\in \mathcal{A}_{W,\rho}(T^*\mathbb{T}^d)$  for some fixed $\rho>0$. Let $(\varepsilon_\hbar)$ be a sequence of positive real numbers satisfying \begin{equation}
\varepsilon_\hbar \ll \hbar.
\end{equation} 
Then, there exists a sequence of integrable counterterms $R_\hbar = R_\hbar(V) \in S_W(\R^d)$ such that $\Vert R_\hbar \Vert_W \lesssim \Vert V \Vert_{W,\rho}$, uniformly in $\hbar \in (0,1]$,  so that \eqref{e:stability_spectrum}, \eqref{e:semiclassical_measures}, and \eqref{e:quantum_limits} hold.
\end{teor}

\subsection*{Acknowledgments} The author thanks Fabricio Macià, Gabriel Rivière and Stéphane Nonnenmacher for useful discussions while preparing this manuscript. This work has been supported by a predoctoral grant from Fundación La Caixa - Severo Ochoa International Ph.D. Program at the Ins\-tituto de Ciencias Matemáticas (ICMAT-CSIC-UAM-UC3M-UCM).

\section{Proof of Theorem \ref{t:main_result_1}}

To prove Theorem \ref{t:main_result_1}, we will use a classical KAM result due to Moser \cite{Mos67} about small perturbations of constant vector fields on the torus. Precisely, we will recall a work of Pöschel \cite{Pos11} that simplifies the KAM iteration argument. On the other hand, we will use Egorov's theorem to establish the classic-quantum correspondence to obtain our result in terms of the quantum system. Our approach is similar to that of Bambusi et. al. in \cite{Bam17}, in which they obtain reducibility for a class of perturbations of the quantum harmonic oscillator.

The proof of Theorem \ref{t:main_result_1} is divided in two parts. First, we prove that the family $\widehat{P}_{\varphi(\omega), \hbar}$ is unitarily equivalent to $\widehat{L}_{\omega, \hbar}$. This shows the stability of the spectrum along this family. The following holds:

\begin{teor}
\label{KAM_Theorem}
Let $s,\rho > 0$ and $V \in \mathscr{L}_{s,\rho}$ be real analytic verifying \eqref{e:small_condition}. Then, there exist a real change of frequencies $\varphi: \Omega \to \Omega_\rho$ satisfying
$$
\sup_{\omega \in \Omega} \vert \varphi(\omega) - \omega \vert \leq C_1 \vert V \vert_{s,\rho},
$$ 
and a family of unitary operators $ \Omega \ni \omega \longmapsto \mathcal{U}_{\omega}$ on $L^2(\mathbb{T}^d)$ such that
\begin{equation}
\mathcal{U}_{\omega} \, \widehat{P}_{\varphi(\omega), \hbar} \, \mathcal{U}_{\omega}^* = \widehat{L}_{\omega,\hbar}.
\end{equation}
\end{teor}

\begin{remark}
In particular, if $V = 0$ then $\varphi = \operatorname{Id}$ and $\mathcal{U}_{\omega} = \operatorname{Id}$.
\end{remark}

In the second part, we will compare the semiclassical measures and quantum limits of $\widehat{P}_{\varphi(\omega),\hbar}$ with those of $\widehat{L}_{\omega,\hbar}$. 

\subsection{A classical KAM theorem}

We first recall the result of Pöschel \cite{Pos11}. We will use the Diophantine property \eqref{e:diophantine_condition2} for the sake of simplicity, but the more general \textit{Rüssmann condition} considered in \cite{Pos11} would be valid as well.

\begin{teor}[\cite{Pos11}]
\label{poschel}
Let $\Omega \subset \R^d$ be a compact set of \textit{strongly nonresonant frequencies}, that is, $\omega \in \Omega$ satisfies \eqref{e:diophantine_condition2}. Let $s,\rho > 0$ and $V \in \mathscr{L}_{s,\rho}$ such that
\begin{equation}
\label{e:small_condition_generalized}
\vert V \vert_{s,\rho} =\varepsilon < \frac{\rho}{16} \leq \frac{\varsigma}{32 \lambda^\gamma},
\end{equation}
where the constants $\varsigma$ and $\gamma$ are defined in \eqref{e:diophantine_condition2}, and $\lambda$ is so large that
\begin{equation}
\label{large_lambda_2}
r := 8 \gamma \left(\frac{ 1 +\log \lambda  }{\lambda} \right) < \frac{s}{2}.
\end{equation}
Then there exists a real map $\varphi : \Omega \to \Omega_\rho$ and, for every $\omega \in \Omega$, a real analytic diffeomorphism $\theta_\omega$ of the $d$-torus such that, denoting
\begin{equation*}
\Theta_\omega(x,\xi) = \big( \theta_\omega(x), [(\partial_x \theta_\omega(x))^T]^{-1} \xi \big),
\end{equation*}
the following holds:
\begin{equation}
\label{e:conjugation}
\big( \mathcal{L}_{\varphi(\omega)} + V(\cdot; \varphi(\omega)) \big) \circ \Theta_\omega  = \mathcal{L}_\omega.
\end{equation}
Moreover, 
\begin{equation}
\label{e:estimates_conjugation}
\sup_{\omega \in \Omega} \vert \varphi(\omega) - \omega \vert \leq  \varepsilon , \quad \sup_{\omega \in \Omega} \sup_{x \in \mathbb{T}^d} \vert \theta_\omega(x) - x \vert \leq  r \, \varsigma^{-1} \lambda^\gamma \varepsilon .
\end{equation}
\end{teor}
This means that, for every $\omega'$ in the Cantor set $\varphi(\Omega)$, the Hamiltonian $P_{\omega'}$ is canonically conjugate to the unperturbed one $\mathcal{L}_\omega$, where $\omega = \varphi^{-1}(w')$, and hence the energy level $P_{\omega'}^{-1}(1)$ is foliated by invariant tori with frequency $\omega$.

Using this result, the proof of Theorem \ref{KAM_Theorem} is straightforward in terms of Egorov's  theorem:

\begin{proof}[Proof of Theorem \ref{KAM_Theorem}]
We define the unitary operator $\mathcal{U}_\omega : L^2(\mathbb{T}^d) \to L^2(\mathbb{T}^d)$ by 
\begin{equation}
\label{d:unitary_operator}
\mathcal{U}_\omega \psi(x) :=  \sqrt{ \vert \det d\theta_\omega(x) \vert} \, \psi \big( \theta_\omega(x) \big) .
\end{equation}
By Egorov's theorem, which is exact in this case, we conclude that
$$
\mathcal{U}_{\omega} \, \widehat{P}_{\varphi(\omega), \hbar} \, \mathcal{U}_{\omega}^* = \Op_\hbar\big( (\mathcal{L}_{\varphi(\omega)} + V(\cdot, \varphi(\omega))) \circ \Theta_\omega  \big) = \Op_\hbar(\mathcal{L}_\omega) = \widehat{L}_{\omega,\hbar}.
$$
\end{proof}

\subsection{Quantum limits and semiclassical measures}

We next complete the proof of Theorem \ref{t:main_result_1}. 

\begin{prop}
\label{p:non_perturbe_operator}
Let $\omega \in \R^d$ be linearly independent over the rationals\footnote{That is, if $k \in \mathbb{Z}^d$ satisfies $\omega \cdot k = 0$ then $k = 0$.}. Then 
$$
\mathcal{M}(\widehat{L}_{\omega,\hbar}) = \bigcup_{\xi \in \mathcal{L}_\omega^{-1}(1)} \big \{ \mathfrak{h}_{\mathbb{T}^d \times \{ \xi \}} \big \} \cup \{0 \}, \quad \mathcal{N}(\widehat{L}_{\omega,\hbar}) = \left \{ \frac{1}{(2\pi)^d}dx \right \}.
$$  
\end{prop}
\begin{proof}

We recall that the point-spectrum of $\widehat{L}_{\omega,\hbar}$ is given by
$$
\operatorname{Sp}^p_{L^2(\mathbb{T}^d)} \big( \widehat{L}_{\omega, \hbar} \big) = \{ \lambda_{k,\hbar} = \hbar \, \omega \cdot k \, : \, k \in \mathbb{Z}^d \}.
$$
Each eigenvalue has multiplicity equal to one due to the nonresonant condition on $\omega$. Moreover, the set of eigenfunctions is just given by
$$
e_k(x) = \frac{e^{ik \cdot x}}{(2\pi)^{d/2}}, \quad k \in \mathbb{Z}^d.
$$
By a direct calculation using identity \eqref{e:wigner_distribution} for the Wigner distribution on the torus, for every test function $a \in \mathcal{C}_c^\infty(T^*\mathbb{T}^d)$:
$$
W^\hbar_{e_k}(a) = \frac{1}{(2\pi)^d} \int_{\mathbb{T}^d} a(x, \hbar k) \, dx, \quad k \in \mathbb{Z}^d.
$$
Equivalently, $W^\hbar_{e_k} = \mathfrak{h}_{\mathbb{T}^d \times \{\hbar k \}}$. Given a sequence
\begin{equation}
\label{eigenvalues}
\lambda_{k_j,\hbar_j} = \hbar_j \, \omega \cdot k_j \to 1, \quad \text{as }\hbar_j \to 0,
\end{equation} 
if $\hbar_j k_j \to \xi \in \R^d$ then clearly $\xi \in \mathcal{L}_\omega^{-1}(1)$. In other words,  $\mathfrak{h}_{\mathbb{T}^d \times \{ \xi \}} \in \mathcal{M}(\widehat{L}_{\omega, \hbar})$. Reciprocally, any point $\xi \in \mathcal{L}_\omega^{-1}(1)$ can be obtained as the limit of a sequence $(\hbar_j k_j)$ satisfying \eqref{eigenvalues}, and hence any measure $\mathfrak{h}_{\mathbb{T}^d \times \{ \xi \}}$ is the semiclassical measure associated to some sequence of eigenfunctions. Finally, since $\mathcal{L}_\omega^{-1}(1)$ is not compact, there are also sequences $(\hbar_j k_j)$ satisfying \eqref{eigenvalues}  such that
$$
\lim_{j \to \infty} \vert \hbar_j k_j \vert = \infty.
$$
For those sequences, we have $\mu = 0$. Thus $0 \in \mathcal{M}(\widehat{L}_{\omega, \hbar})$.

The second assertion is trivial since
$$
\vert e_k(x) \vert^2 = \frac{1}{(2\pi)^d}, \quad k \in \mathbb{Z}^d.
$$
\end{proof}

We are now in position to conclude the proof of Theorem \ref{t:main_result_1}. Since 
$$
\widehat{L}_{\omega,\hbar} = \mathcal{U}_{\omega} \, \widehat{P}_{\varphi(\omega), \hbar} \, \mathcal{U}_{\omega}^*,
$$
where $\mathcal{U}_{\omega}$ is unitary on $L^2(\mathbb{T}^d)$, the spectrum of $\widehat{P}_{\varphi(\omega), \hbar}$ is the same as the spectrum of $\widehat{L}_{\omega,\hbar}$, and the eigenfunctions are precisely
$$
\Psi_{k} =  \mathcal{U}_{\omega}^* e_k, \quad k \in \mathbb{Z}^d.
$$
Thus, applying Egorov's theorem,
$$
W^\hbar_{\Psi_{k}}(a) = W^\hbar_{e_k}( a \circ \Theta_\omega) + O(\hbar), \quad a \in \mathcal{C}_c^\infty(T^*\mathbb{T}^d),
$$
and similarly, using \eqref{d:unitary_operator},
\begin{equation}
\label{e:egorov_for_quantum_limits}
\int_{\mathbb{T}^d} b(x) \vert \Psi_{k}(x) \vert^2 dx =  \int_{\mathbb{T}^d} b \circ \theta_\omega(x) \vert e_k(x) \vert^2 dx, \quad b \in \mathcal{C}^\infty(\mathbb{T}^d).
\end{equation}
Then the proof of Theorem \ref{t:main_result_1} reduces to the proof of Proposition \ref{p:non_perturbe_operator}.

\section{Proof of Theorems \ref{t:main_result_2} and \ref{t:main_result_3}}

The main ingredient in the proof of Theorem \ref{t:main_result_2} is the following quantum version of the renormalization problem:
\begin{teor}
\label{KAM_theorem_precise}
Let $\omega \in \R^d$ be a strongly non resonant frequency satisfying \eqref{e:diophantine_condition2}, and let $V \in \mathcal{A}_{s,\rho}(T^*\mathbb{T}^d)$  for some fixed $s,\rho>0$. Assume that $\varepsilon_\hbar = \hbar$, and
\begin{equation}
\label{e:small_condition3}
\Vert V \Vert_{s,\rho} \leq \frac{\varsigma}{64}\left( \frac{\sqrt{\rho}}{2(\gamma - 1)} \right)^{2(\gamma-1)}.
\end{equation}
Then there exist a sequence of unitary operators $\mathcal{U}_\hbar : L^2(\mathbb{T}^d) \to L^2(\mathbb{T}^d)$, and a sequence of counterterms $R_\hbar \in \mathcal{A}_{s}(\R^d)$ such that \begin{equation}
\label{L2_conjugation2}
\mathcal{U}_\hbar \, \big( \widehat{L}_{\omega,\hbar} + \varepsilon_\hbar \Op_\hbar(V - R_\hbar) \big)\,  \mathcal{U}_\hbar^* = \widehat{L}_{\omega,\hbar}.
\end{equation}
Moreover,
$$
\Vert R_\hbar \Vert_{\mathcal{A}_s(\R^d)} \leq 2 \Vert V \Vert_{s,\rho}, \quad \forall \hbar \in (0,1].
$$ 
\end{teor}

In the case $\varepsilon_\hbar \ll \hbar$, we can remove condition \eqref{e:small_condition3}; we will prove Theorem \ref{t:main_result_3} applying the following version of Theorem \ref{KAM_theorem_precise} with less regularity:

\begin{teor}
\label{KAM_theorem_precise_1}
Let $\omega \in \R^d$ be a strongly non resonant frequency satisfying \eqref{e:diophantine_condition2}, and let $V \in \mathcal{A}_{W,\rho}(T^*\mathbb{T}^d)$  for some fixed $\rho>0$. Assume that $\varepsilon_\hbar \ll \hbar$. Then there exist a sequence of unitary operators $\mathcal{U}_\hbar : L^2(\mathbb{T}^d) \to L^2(\mathbb{T}^d)$, and a sequence of counterterms $R_\hbar \in S_W(\R^d)$ such that
\eqref{L2_conjugation2} holds.
Moreover,
$$
\Vert R_\hbar \Vert_W \leq 2 \Vert V \Vert_{W,\rho}, \quad \forall \hbar \in (0,1].
$$ 
\end{teor}

\subsection{KAM iterative algorithm}\label{KAM_process}

To find $R_\hbar$, we start from the full renormalized ope\-rator $\widehat{Q}_\hbar$ with $R_\hbar$ as unknown and then we will construct $\mathcal{U}_\hbar$ and $R_\hbar$ by an iterative  averaging method. We will find the renormalization function $R_\hbar$ as an infinite sum of the form
$$
R_\hbar := \sum_{j=1}^\infty R_{j,\hbar},
$$
where each $R_{j,\hbar}$ will be  determined at each step of the iteration and the sum will be proven to converge. 

We initially set $V_1 := V$, and consider
\begin{equation}
\label{e:first_hamiltonian}
\widehat{Q}_{1,\hbar} := \widehat{Q}_\hbar = \widehat{L}_{\omega,\hbar} + \varepsilon_\hbar \left( \Op_\hbar(V_1) - \sum_{j=1}^\infty \Op_\hbar(R_{j,\hbar}) \right).
\end{equation}
The goal at the first step of the iteration is to choose a good first term $R_{1,\hbar}$, then average the term $V_1$ by the flow generated by $\mathcal{L}_{\omega}$, and finally estimate the remainder terms.  Given $a \in \mathcal{C}^\infty(T^*\mathbb{T}^d)$ we define its average $\langle a \rangle$ along the flow 
$$
\phi_t^{\mathcal{L}_\omega} \, : \, (x,\xi) \mapsto (x + t\omega, \xi),
$$
by the following limit in the $\mathcal{C}^\infty$-topology of $T^*\mathbb{T}^d$:
\begin{equation}
\label{e:averaged_symbol}
\langle a \rangle(\xi) := \lim_{T \to \infty} \frac{1}{T} \int_0^T a \circ \phi_t^{\mathcal{L}_\omega}(x,\xi) \, dt  = \frac{1}{(2\pi)^d} \int_{\mathbb{T}^d} a (x,\xi) dx = \frac{1}{(2\pi)^{d/2}} \, \widehat{a}(0,\xi),
\end{equation}
where recall that we have used the convention for the Fourier coefficients of $a$,
$$
\widehat{a}(k,\xi) := \big \langle a , e_k \big \rangle_{L^2(\mathbb{T}^d)} = \frac{1}{(2\pi)^{d/2}} \int_{\mathbb{T}^d} a(x,\xi) e^{-i k \cdot x} dx, \quad k \in \mathbb{Z}^d.
$$
If $a$ is bounded together with all its derivatives, Egorov's theorem allows us to define the quantum average of $\Op_\hbar(a)$ by
\begin{equation}
\label{e:quatum_average}
\langle \Op_\hbar(a) \rangle := \lim_{T \to \infty}\frac{1}{T} \int_0^T e^{\frac{i t}{\hbar} \widehat{L}_{\omega,\hbar}} \, \Op_\hbar(a) \, e^{-\frac{i t}{\hbar} \widehat{L}_{\omega,\hbar}} \, dt,
\end{equation}
and, since $\mathcal{L}_\omega$ is a polynomial of degree one, we have that the limit exists in the strong-operator norm, and
$$
\langle \Op_\hbar(a) \rangle = \Op_\hbar(\langle a \rangle).
$$
We  set $R_{1,\hbar} := \langle V_1 \rangle$ and consider a unitary operator of the form
$$
U_{1,\hbar}(t) := e^{\frac{i t \varepsilon_\hbar}{\hbar} \Op_\hbar(F_1)} = \sum_{j=0}^\infty \frac{1}{j!} \left( \frac{i t\varepsilon_\hbar}{\hbar} \right)^j \Op_\hbar(F_1)^j, \quad t \in [0,1],
$$
where $\Op_\hbar(F_1)$ will be the solution of the cohomological equation
\begin{equation}
\label{e:cohomological_equation_1}
\frac{i}{\hbar} [ \widehat{L}_{\omega,\hbar} , \Op_\hbar(F_1)] =  \Op_\hbar(V_1 -  R_1 ), \quad \langle F_1 \rangle = 0.
\end{equation}
We show in  Lemma \ref{l:solution_cohomological_equation} below  how to solve this cohomological equation. Moreover, the Diophantine condition \eqref{e:diophantine_condition2} on $\omega$ will allow us to bound the solution $F_1$ with a bit of loss of analyticity in the variable $x$. Denoting now $U_{1,\hbar} = U_{1,\hbar}(1)$ and $\widehat{Q}_{2,\hbar} := U_{1,\hbar} \, \widehat{Q}_{1,\hbar} \, U_{1,\hbar}^*$, and using Taylor's theorem, we get
\begin{align*}
\widehat{Q}_{2,\hbar} & =   \widehat{L}_{\omega, \hbar} + \frac{i \varepsilon_\hbar}{\hbar}[\Op_\hbar(F_1) , \widehat{L}_{\omega,\hbar} ] + \varepsilon_\hbar \Op_\hbar(V_1-R_1)  \\[0.2cm]
 & \quad  +  \left(\frac{i \varepsilon_\hbar}{\hbar} \right)^2 \int_0^1(1-t) U_{1,\hbar}(t) [\Op_\hbar(F_1), [\Op_\hbar(F_1) , \widehat{L}_{\omega,\hbar} ]] U_{1,\hbar}(t)^* dt   \\[0.2cm]
 & \quad + \frac{i\varepsilon_\hbar^2}{\hbar} \int_0^1 U_{1,\hbar}(t) [\Op_\hbar(F_1) , \Op_\hbar(V_1 - R_1)] U_{1,\hbar}(t)^* dt  \\[0.2cm]
 & \quad  - \varepsilon_\hbar \sum_{j=2}^\infty U_{1,\hbar} \Op_\hbar(R_{j,\hbar}) U_{1,\hbar}^*.
\end{align*}
With this and the cohomological equation \eqref{e:cohomological_equation_1}, we obtain
$$
\widehat{Q}_{2,\hbar} = \widehat{L}_{\omega,\hbar} + \varepsilon_\hbar \left(\Op_\hbar(V_{2,\hbar}) - \sum_{j=2}^\infty U_{1,\hbar} \Op_\hbar(R_{j,\hbar}) U_{1,\hbar}^* \right),
$$
where
\begin{align}
\label{a:V_2}
\Op_\hbar(V_{2,\hbar}) & = \frac{i\varepsilon_\hbar}{\hbar} \int_0^1 t U_{1,\hbar}(t)[ \Op_\hbar(F_1),  \Op_\hbar(V_1 - R_1) ]U_{1,\hbar}(t)^* dt.
\end{align}

Now we proceed to explain the induction step. Assume we have constructed unitary operators $U_{1,\hbar}, \ldots , U_{n-1,\hbar}$ and counterterms $R_{1,\hbar}, \ldots , R_{n-1,\hbar}$ so that
\begin{equation}
\label{e:n-1_step}
\widehat{Q}_{n,\hbar} =  U_{n-1,\hbar} \cdots U_{1,\hbar} \,  \widehat{Q}_{1,\hbar} \, U_{1,\hbar}^* \cdots U_{n-1,\hbar}^* =  \widehat{L}_{\omega,\hbar} + \varepsilon_\hbar \left( \Op_\hbar(V_{n,\hbar}) - \sum_{j=n}^\infty \widehat{E}_{n,j,\hbar} \right),
\end{equation}
where, for every $j \geq n$:
\begin{align*}
\widehat{E}_{n,j,\hbar} = \Op_\hbar(E_{n,j,\hbar}) & := U_{n-1,\hbar} \cdots U_{1,\hbar} \Op_\hbar(R_{j,\hbar}) U_{1,\hbar}^* \cdots U_{n-1,\hbar}^*.
\end{align*}
We will choose $R_{n,\hbar}$ to be the unique solution of the operator equation
\begin{equation}
\label{e:equation_to_find_R_n}
\langle \widehat{E}_{n,n,\hbar} \rangle = \langle U_{n-1,\hbar} \cdots U_{1,\hbar} \Op_\hbar(R_{j,\hbar}) U_{1,\hbar}^* \cdots U_{n-1,\hbar}^* \rangle = \langle \Op_\hbar (V_{n,\hbar} )\rangle,
\end{equation}
given by Lemma \ref{l:inverse_map} below. We next consider the unitary operator
$$
U_{n,\hbar}(t) := e^{\frac{it\varepsilon_\hbar}{\hbar} \Op_\hbar(F_{n,\hbar})} = \sum_{j=0}^\infty \frac{1}{j!} \left( \frac{it\varepsilon_\hbar}{\hbar} \right)^j \Op_\hbar(F_{n,\hbar})^j, \quad t \in [0,1],
$$
where $\Op_\hbar(F_{n,\hbar})$ solves the cohomological equation (see Lemma \ref{l:solution_cohomological_equation}):
\begin{equation}
\label{e:cohomological_equation_n}
\frac{i}{\hbar} [ \widehat{L}_{\omega,\hbar} , \Op_\hbar(F_{n,\hbar})] =  \Op_\hbar(V_{n,\hbar} - E_{n,n,\hbar}), \quad \langle F_{n,\hbar} \rangle = 0.
\end{equation}
As in the first step, we denote $U_{n,\hbar} := U_{n,\hbar}(1)$. Defining $\widehat{Q}_{n+1,\hbar} := U_{n,\hbar} \, \widehat{Q}_{n,\hbar} \, U_{n,\hbar}^*$, we use Taylor's theorem to expand 
\begin{align*}
\widehat{Q}_{n+1,\hbar} & =  \widehat{L}_{\omega, \hbar} + \frac{i\varepsilon_\hbar}{\hbar}[\Op_\hbar(F_{n,\hbar}) , \widehat{L}_{\omega,\hbar} ]  + \varepsilon_\hbar \Op_\hbar(V_{n,\hbar}- E_{n,n,\hbar}) \\[0.2cm]
 & \quad  +  \left(\frac{i\varepsilon_\hbar}{\hbar} \right)^2 \int_0^1(1-t) U_{n,\hbar}(t) [\Op_\hbar(F_{n,\hbar}), [\Op_\hbar(F_{n,\hbar}) , \widehat{L}_{\omega,\hbar} ]] U_{n,\hbar}(t)^* dt \\[0.2cm]
 & \quad  + \frac{i\varepsilon_\hbar^2}{\hbar} \int_0^1 U_{n,\hbar}(t) [\Op_\hbar(F_{n,\hbar}) , \Op_\hbar(V_{n,\hbar} - E_{n,n,\hbar})] U_{n,\hbar}(t)^* dt  \\[0.2cm]
 & \quad  - \varepsilon_\hbar \sum_{j=n+1}^\infty U_{n,\hbar} \Op_\hbar(E_{n,j,\hbar}) U_{n,\hbar}^*,
\end{align*}
and using  the cohomological equation \eqref{e:cohomological_equation_n}, we obtain
$$
\widehat{Q}_{n+1,\hbar} = \widehat{L}_{\omega,\hbar} + \varepsilon_\hbar \left( \Op_\hbar(V_{n+1,\hbar}) - \sum_{j=n+1}^\infty \Op_\hbar(E_{n+1,j,\hbar}) \right),
$$
where
\begin{align}
\label{e:generalV}
\Op_\hbar(V_{n+1,\hbar}) & = \frac{i\varepsilon_\hbar}{\hbar} \int_0^1 t U_{n,\hbar}(t)[ \Op_\hbar(F_n), \Op_\hbar(V_{n,\hbar} - E_{n,n,\hbar} ) ] U_{n,\hbar}(t)^*dt, 
\end{align}
and, for every $j \geq n+1$,
$$
\widehat{E}_{n+1,j,\hbar} = \Op_\hbar(E_{n+1,j,\hbar}) := U_{n,\hbar} \Op_\hbar(E_{n,j,\hbar}) U_{n,\hbar}^*.
$$

This iteration procedure will converge provided that $V$ is sufficiently small. Precisely, we will obtain a unitary operator $\mathcal{U}_\hbar$ as the limit, in the strong operator $\mathcal{L}(L^2)$-norm,
$$
\mathcal{U}_\hbar := \lim_{n\to \infty} U_{n,\hbar} \cdots U_{1,\hbar}.
$$

\subsection{Cohomological equations}

In this section we explain how to solve the cohomological equations appearing in our averaging method. This is a standard technique when dealing with small divisors problems.

\begin{lemma}
\label{l:solution_cohomological_equation}
Let $V \in \mathcal{A}_{s,\rho}(T^*\mathbb{T}^d)$. Then, the cohomological equation
\begin{equation}
\label{e:cohomological_equation_lemma}
\frac{i}{\hbar} [ \widehat{L}_{\omega,\hbar} , \Op_\hbar(F)] = \Op_\hbar(V - \langle V \rangle), \quad \langle F \rangle = 0,
\end{equation}
has a unique solution $F \in \mathcal{A}_{s,\rho-\sigma}(T^*\mathbb{T}^d)$, for every $0 < \sigma < \rho$, such that
\begin{equation}
\label{e:bound_solution_cohomologial}
\Vert F \Vert_{s, \rho-\sigma} \leq \varsigma^{-1} \left( \frac{\gamma - 1}{e\sigma} \right)^{\gamma -1} \Vert V \Vert_{s,\rho}.
\end{equation}
Similarly, if $V \in \mathcal{A}_{W,\rho}(\R^{2d})$, then there exists a unique $F \in \mathcal{A}_{W,\rho-\sigma}(T^*\mathbb{T}^d)$ solving \eqref{e:cohomological_equation_lemma}, and \eqref{e:bound_solution_cohomologial} holds replacing $\Vert \cdot \Vert_{s,\cdot}$ by $\Vert \cdot \Vert_{W,\cdot}$.
\end{lemma}

\begin{proof}
Using the properties of the symbolic calculus for the Weyl quantization, which in this case is exact since $\mathcal{L}_\omega$ is a polynomial of degree one, equation \eqref{e:cohomological_equation_lemma} at symbol level is just
\begin{equation}
\label{e:symbolic_cohomological}
\{ \mathcal{L}_\omega, F \} = V - \langle V \rangle, \quad \langle F \rangle = 0.
\end{equation}
Recall also that, by \eqref{e:averaged_symbol}, the average of $V$ is given by
$$
\langle V \rangle(\xi) = \frac{1}{(2\pi)^d} \int_{\mathbb{T}^d} V(x,\xi) dx =  \frac{1}{(2\pi)^{d/2}} \widehat{V} (0,\xi).
$$
On the other hand, since
$$
\{ \mathcal{L}_\omega, F \}(x,\xi) =  \sum_{k \in \mathbb{Z}^d} i \omega \cdot k \,  \widehat{F}(k,\xi) e_k(x),
$$
we obtain the following formal series for the solution of \eqref{e:symbolic_cohomological}:
\begin{equation}
\label{e:solution_cohomological_equation}
F(x,\xi) =  \sum_{k \in \mathbb{Z}^d\setminus \{ 0 \}} \frac{\widehat{V}(k,\xi)}{i\omega \cdot k}  e_k(x).
\end{equation}
Finally, by Diophantine condition \eqref{e:diophantine_condition2} and estimate \eqref{e:elementary_estimate}, we conclude that
\begin{equation}
\label{e:cohomological_estimate}
\Vert F \Vert_{s,\rho-\sigma} \leq \varsigma^{-1} \left( \frac{\gamma - 1}{e\sigma} \right)^{\gamma -1} \Vert V \Vert_{s,\rho}.
\end{equation}
 Since the loss of analyticity of $F$ with respect to $V$ occurs only in the variable $x$, one can substitute the norms $\Vert \cdot \Vert_{s,\cdot}$ by $\Vert \cdot \Vert_{W,\cdot}$ in \eqref{e:cohomological_estimate} to obtain also the second assertion of the statement.
\end{proof}

\subsection{Convergence}


We next show that the algorithm sketched in Section \ref{KAM_process} converges under appropriate hypothesis. 

\begin{proof}[Proof of Theorem \ref{KAM_theorem_precise}]
We start by fixing the following universal constants:
\begin{equation}
\label{e:constants}
\alpha := \frac{1}{4}, \quad \beta  := \frac{1}{16}, \quad \lambda := e^{\frac{\beta}{1- \sqrt{\alpha}}} - 1.
\end{equation}
Now set
\begin{equation}
\label{e:set_constants}
\rho_1 := \rho, \quad \sigma_1 := \frac{\rho}{2e(\gamma - 1)} \alpha^{\frac{1}{2(\gamma-1)}}.
\end{equation}
By Lemma \ref{l:solution_cohomological_equation}, \eqref{e:set_constants} and hypothesis \eqref{e:small_condition3},
\begin{equation*}
\label{e:step_one1}
\Vert F_1 \Vert_{s,\rho_1 - \sigma_1} \leq \varsigma^{-1} \left( \frac{\gamma - 1}{e \sigma_1} \right)^{\gamma -1} \Vert V_1 \Vert_{s,\rho_1} \leq \frac{\beta}{2}.
\end{equation*}
Then, using \eqref{a:V_2} and the conventions of Appendix \ref{s:tools_of_analytic_calculus},
$$
V_{2,\hbar} = \frac{i \varepsilon_\hbar}{\hbar} \int_0^1 t \Psi_{t,\hbar}^{\varepsilon_\hbar F_1} \big( [F_1, V_1 -R_1]_\hbar \big)dt,
$$ 
where $[\cdot , \cdot]_\hbar$ and $\Psi_{t,\hbar}^F$ are defined in \eqref{e:commutator_symbol} and \eqref{e:conjugation_symbol}, estimate \eqref{trivial_bound}, and Lemma \ref{l:bad_flow_estimate}, we obtain
\begin{equation*}
\label{e:step_one2}
\Vert V_2 \Vert_{s,\rho_1 - \sigma_1} \leq \beta(1 + \beta) \Vert V_1 \Vert_{s,\rho_1} \leq \alpha \Vert V_1 \Vert_{s,\rho_1}.
\end{equation*}
Moreover, by definition of $R_1$ and of our analytic spaces (Definition \ref{Banach_spaces_of_symbols}):
\begin{equation*}
\label{e:step_one3}
\Vert R_1 \Vert_{\mathcal{A}_s(\R^d)} = \Vert \langle V_1 \rangle \Vert_{\mathcal{A}_s(\R^d)} \leq \Vert V_1 \Vert_{s,\rho_1}.
\end{equation*}
This shows the starting step of the induction. We next define sequences
$$
\sigma_{n+1} := \sigma_n \alpha^{\frac{1}{2(\gamma-1)}},  \quad \rho_{n+1} := \rho_n - \sigma_n, \quad n \geq 1, 
$$
and assume the following induction hypothesis: for every $n \geq 2$ and $1 \leq j \leq n-1$, there exist $F_{j,\hbar}$ and $R_{j,\hbar}$ so that
\begin{equation}
\label{e:induction_hypothesis}
\Vert F_{j,\hbar} \Vert_{s,\rho_j} \leq \frac{\beta \alpha^{\frac{j-1}{2}}}{2}, \quad \Vert R_{j,\hbar} \Vert_{s} \leq \frac{\alpha^{j-1}}{1-\lambda} \Vert V_1 \Vert_{s,\rho_1},
\end{equation}
and, for $V_{n,\hbar}$ obtained in $\eqref{e:n-1_step}$, 
\begin{equation}
\label{e:induction_hypothesis2}
\Vert V_{n,\hbar} \Vert_{s, \rho_n} \leq \alpha^{n-1} \Vert V_1 \Vert_{s, \rho_1}.
\end{equation}
To prove the induction step, we first recall that, for every $j \geq n$, $\widehat{E}_{n,j,\hbar} = \Op_\hbar(E_{n,j,\hbar})$, where
$$
E_{n,j,\hbar} = \Psi_{1,\hbar}^{\varepsilon_\hbar F_{n-1}} \circ \cdots \circ \Psi_{1,\hbar}^{\varepsilon_\hbar F_{1}} R_{j,\hbar}.
$$
Our choice of $R_{n,\hbar}$ is the unique solution of equation \eqref{e:equation_to_find_R_n}. At symbol level, equation \eqref{e:equation_to_find_R_n} reads
\begin{equation}
\label{e:special_equation}
\langle E_{n,n,\hbar} \rangle = \langle \Psi_{1,\hbar}^{\varepsilon_\hbar F_{n-1}} \circ \cdots \circ \Psi_{1,\hbar}^{\varepsilon_\hbar F_{1}} R_{n,\hbar} \rangle = \langle V_{n,\hbar} \rangle {\color{red}.}
\end{equation} 
The solution exists and is unique in view of the following:

\begin{lemma}
\label{l:inverse_map}
Assume that $\varepsilon_\hbar \leq \hbar$. Let $\langle V \rangle \in \mathcal{A}_s(\R^d)$ and let $F_j \in  \mathcal{A}_{s,\rho_j}(T^*\mathbb{T}^d)$ for $1 \leq j \leq  n-1$ and some positive numbers $\rho_1 \geq \cdots \geq  \rho_{n-1} > 0$ such that
$$
2 \Vert F_j \Vert_{s,\rho_j-\sigma_j} \leq \beta \, \alpha^{\frac{j-1}{2}},
$$
where $\alpha, \beta > 0$ satisfy
$$
\lambda := e^{\frac{\beta}{1- \sqrt{\alpha}}} - 1 < 1.
$$
Then, there exists $R \in \mathcal{A}_s(\R^d)$ so that
$$
\langle \Psi_{1,\hbar}^{\varepsilon_\hbar  F_{n-1}} \circ \cdots \circ \Psi_{1,\hbar}^{\varepsilon_\hbar F_1} R \rangle = \langle V \rangle,
$$
and
$$
\Vert R \Vert_{\mathcal{A}_s(\R^d)} \leq \frac{1}{1 - \lambda} \Vert \langle V \rangle \Vert_{\mathcal{A}_s(\R^d)}, \quad \Vert \Psi_{1,\hbar}^{\varepsilon_\hbar  F_{n-1}} \circ \cdots \circ \Psi_{1,\hbar}^{\varepsilon_\hbar F_1} R \Vert_{s,\rho_n} \leq \frac{1 + \lambda}{1-\lambda} \Vert \langle V \rangle \Vert_{\mathcal{A}_s(\R^d)}.
$$
\end{lemma}

\begin{proof}
Define the map $T : \mathcal{A}_s(\R^d) \to \mathcal{A}_s(\R^d)$ by
$$
T(R) := \langle \Psi_{\hbar,1}^{\varepsilon_\hbar F_{n-1}} \circ \cdots \circ \Psi_{\hbar,1}^{\varepsilon_\hbar F_1} R \rangle.
$$
By Lemma \ref{l:bad_flow_estimate}, we have
\begin{align*}
\Vert T(R) - R \Vert_{\mathcal{A}_s(\R^d)} & \\[0.2cm]
&  \hspace*{-2cm} \leq \left[ \prod_{j=1}^{ n-1} (1 + \beta \alpha^{\frac{j-1}{2}}) - 1\right] \Vert R \Vert_{\mathcal{A}_s(\R^d)} \leq \Big( e^{\frac{\beta}{1- \sqrt{\alpha}}} - 1\Big) \Vert R \Vert_{\mathcal{A}_s(\R^d)} = \lambda \Vert R \Vert_{\mathcal{A}_s(\R^d)}.
\end{align*}
Then, there exists an inverse map $T^{-1} : \mathcal{A}_s(\R^d) \to \mathcal{A}_s(\R^d)$ defined by Neumann series, and
$$
\Vert T^{-1} \Vert_{\mathcal{A}_s \to \mathcal{A}_s} \leq \frac{1}{1- \lambda}.
$$
Finally, applying Lemma \ref{l:bad_flow_estimate} one more time, we obtain: 
$$
 \Vert \Psi_{\hbar,1}^{\varepsilon_\hbar  F_{n-1}} \circ \cdots \circ \Psi_{\hbar,1}^{\varepsilon_\hbar F_1} R \Vert_{s,\rho_n} \leq \frac{1 + \lambda}{1-\lambda} \Vert \langle V \rangle \Vert_{\mathcal{A}_s(\R^d)}.
$$
This concludes the proof of the Lemma.
\end{proof}

Applying this Lemma to \eqref{e:special_equation}, we obtain
\begin{align*}
\Vert R_{n,\hbar} \Vert_{s,\rho_n} & \leq \frac{1}{1-\lambda} \Vert V_{n,\hbar} \Vert_{s,\rho_n} \leq \frac{\alpha^{n-1}}{1-\lambda} \Vert V_1 \Vert_{s,\rho_1}, \\[0.2cm]
\Vert E_{n,n,\hbar} \Vert_{s,\rho_n} & \leq \frac{1 + \lambda}{1 - \lambda} \Vert V_{n,\hbar} \Vert_{s,\rho_n} \leq  \frac{1+\lambda}{1-\lambda} \alpha^{n-1} \Vert V_1 \Vert_{s,\rho_1}.
\end{align*}
Note that, with our choice of constants \eqref{e:constants}: 
$$
\beta(1+\beta) \left(1 + \frac{1+\lambda}{1-\lambda} \right) \leq \alpha.
$$
We next observe that, by Lemma \ref{l:solution_cohomological_equation} and hypothesis  \eqref{e:small_condition3}, there exists $F_{n,\hbar}$ solving \eqref{e:cohomological_equation_n} such that:
\begin{align*}
\Vert F_{n,\hbar} \Vert_{s,\rho_n - \sigma_n} & \leq \varsigma^{-1} \left( \frac{\gamma - 1}{e \sigma_{n}} \right)^{\gamma -1} \Vert V_n -E_{n,n,\hbar} \Vert_{s, \rho_n} \\[0.2cm]
 & \leq \varsigma^{-1} \left( \frac{\gamma - 1}{e \sigma_{1}} \right)^{\gamma -1} \left( 1+ \frac{1+\lambda}{1-\lambda} \right)\alpha^{\frac{n-1}{2}} \Vert V_1 \Vert_{s,\rho_1} \\[0.2cm]
 & \leq \frac{\beta \alpha^{\frac{n-1}{2}}}{2}.
\end{align*}
Then, recalling \eqref{e:generalV}, which at symbol level reads
$$
V_{n+1,\hbar} = \frac{i \varepsilon_\hbar}{\hbar} \int_0^1 t \Psi_{t,\hbar}^{\varepsilon_\hbar F_n} \big( [F_n, V_{n,\hbar} - E_{n,n,\hbar}]_\hbar \big)dt,
$$
we can apply estimate \eqref{trivial_bound} and Lemmas \ref{l:bad_flow_estimate} and \ref{l:inverse_map} to obtain:
\begin{align*}
\Vert V_{n+1,\hbar} \Vert_{s, \rho_n-\sigma_n} & \leq \beta(1+\beta) \left(1 + \frac{1+\lambda}{1-\lambda} \right) \Vert V_{n,\hbar} \Vert_{s, \rho_n} \\[0.2cm]
 & \leq   \alpha \Vert V_{n,\hbar} \Vert_{s,\rho_n} \leq \alpha^n \Vert V_1 \Vert_{s,\rho_1}. 
\end{align*}
This finishes the induction step. Note that our choice of constants also ensures that
$$
\sum_{n=1}^\infty \sigma_n = \sigma_1 \sum_{j=0}^\infty \left( \frac{1}{2} \right)^{\frac{j}{\gamma-1}} \leq \frac{\rho}{2e (\gamma - 1)} \frac{1}{\log 2^{\frac{1}{\gamma-1}}} \leq \frac{\rho}{2e \log 2} \leq \frac{\rho}{2}.
$$
Moreover,
$$
\Vert R_\hbar \Vert_{\mathcal{A}_s(\R^d)} \leq \sum_{j=1}^\infty \Vert R_{j,\hbar} \Vert_{\mathcal{A}_s(\R^d)} \leq \left( \frac{1}{1- \lambda} \sum_{j=0}^\infty \alpha^j \right)\Vert V_1 \Vert_{s,\rho}  \leq 2 \Vert V_1 \Vert_{s,\rho}.
$$

It remains to show that there exists a unitary operator $\mathcal{U}_\hbar$ so that
$$
\mathcal{U}_\hbar := \lim_{n\to \infty} U_{n,\hbar} \cdots U_{1,\hbar}.
$$
For every $1 \leq n$, we set the unitary operator $\mathcal{U}_{n,\hbar}$ by
$$
\mathcal{U}_{n,\hbar} := U_{n,\hbar} \cdots U_{1,\hbar}.
$$
We have, for every $p \geq 1$:
$$
\mathcal{U}_{n+p,\hbar} -\mathcal{U}_{n,\hbar} = \mathcal{U}_{n,\hbar} \, \mathscr{R}_\hbar(n,p),
$$
where
$$
 \mathscr{R}_h(n,p) :=  e^{\frac{i\varepsilon_\hbar}{\hbar} \widehat{F}_{n+1,\hbar}} \cdots e^{\frac{i\varepsilon_\hbar}{\hbar} \widehat{F}_{n+p,\hbar}} - I, \quad \widehat{F}_{j,\hbar} := \Op_\hbar(F_{j,\hbar}).
$$
By Taylor's theorem, we can write
$$
e^{\frac{i\varepsilon_\hbar}{\hbar} \widehat{F}_{j,\hbar}} = I +  \widehat{B}_{j,\hbar}, \quad  \widehat{B}_{j,\hbar} := \frac{i\varepsilon_\hbar}{\hbar} \widehat{F}_{j,\hbar} \int_0^1 e^{\frac{it\varepsilon_\hbar}{\hbar}  \widehat{F}_{j,\hbar}} \, dt.
$$
Moreover, Lemma \ref{l:analytic_calderon_vaillancourt_torus} and \eqref{e:induction_hypothesis} allow us to bound the $\mathcal{L}(L^2)$ norm of $ \widehat{B}_{j,\hbar}$ by:  
$$
\Vert  \widehat{B}_{j,\hbar} \Vert_{\mathcal{L}(L^2)} \leq \frac{C_{d,\rho} \beta \alpha^{\frac{j-1}{2}}}{2} .
$$
Thus
\begin{align*}
\Vert \mathscr{R}_\hbar(n,p) \Vert_{\mathcal{L}(L^2)} \leq -1 + \prod_{j=1}^p \big( 1 + \Vert  \widehat{B}_{n+j,\hbar} \Vert_{\mathcal{L}(L^2)} \big) \leq -1 + \exp \left[ \frac{ C_{d,\rho} \beta \alpha^{\frac{n-1}{2}}}{2(1-\alpha^{1/2})} \right] .
\end{align*}
Finally, taking the limit $n \to \infty$, we obtain that the sequence $\{ \mathcal{U}_{n,\hbar} \}_{n \geq 1}$ is a Cauchy sequence in the operator norm, and then the result holds.
\end{proof}

\begin{proof}[Proof of Theorem \ref{KAM_theorem_precise_1}]
The proof of Theorem \ref{KAM_theorem_precise_1} follows the same lines of the proof of Theorem \ref{KAM_theorem_precise}, after replacing respectively the spaces $\mathcal{A}_{s,\rho}(T^*\mathbb{T}^d)$ and $\mathcal{A}_s(\R^d)$ by $\mathcal{A}_{W,\rho}(T^*\mathbb{T}^d)$ and $S_W(\R^d)$. Notice that condition \eqref{e:small_condition3} is not required since, by the hypothesis $\varepsilon_\hbar \ll \hbar$,  it is satisfied for $\frac{\varepsilon_\hbar}{\hbar} \Vert V \Vert_{W,\rho}$ instead of $\Vert V \Vert_{s,\rho}$ if $\hbar$ is sufficiently small. Observe also that Lemma \ref{l:bad_flow_estimate} remains valid in view of Corollary \ref{c:same_for_sjostrand}; estimate \eqref{trivial_bound} can be replaced by  \eqref{trivial_bound2}, and one can use estimate \eqref{e:calderon_vaillancourt_2} instead of Lemma \ref{l:analytic_calderon_vaillancourt_torus}.
\end{proof}

\subsection{Semiclassical measures and quantum limits}\label{last_section}

We next complete the proof of Theorem \ref{t:main_result_2}. We will require the following two lemmas:
\begin{lemma}
\label{l:O_lemma}
Let $s,\rho > 0$. Assume that $\varepsilon_\hbar = \hbar$ and $V \in \mathcal{A}_{s,\rho}(T^*\mathbb{T}^d)$ satisfies \eqref{e:small_condition3}. Then, for every $a \in \mathcal{A}_{s,\rho}(T^*\mathbb{T}^d)$,
\begin{equation}
\label{e:first_O}
\Vert  \mathcal{U}_\hbar  \Op_\hbar(a) \,  \mathcal{U}_\hbar^* - \Op_\hbar(a) \Vert_{\mathcal{L}(L^2)} = O(\varepsilon_\hbar),
\end{equation}
and similarly, for every $b \in \mathcal{A}_\rho(\mathbb{T}^d)$, 
\begin{equation}
\label{second_O}
\Vert  \mathcal{U}_\hbar  \Op_\hbar(b) \, \mathcal{U}_\hbar^* - \Op_\hbar(b) \Vert_{\mathcal{L}(L^2)} = O(\varepsilon_\hbar).
\end{equation}
\end{lemma}

\begin{proof}
For every $n \geq 1$, we define:
$$
\delta_n := \left( \frac{1}{2} \right)^{\frac{n-1}{3}} \delta_1, \quad  \delta_1 :=  \min \left \lbrace  \frac{\rho}{10}, \frac{s}{10} \right \rbrace.
$$
Note that
$$
\sum_{n=1}^\infty \delta_n \leq  \min \left \lbrace \frac{\rho}{2}, \frac{s}{2} \right \rbrace.
$$
By \eqref{e:induction_hypothesis}, we have 
$$
\Vert F_{n,\hbar} \Vert_{s, \rho_n} \leq C_{\rho} \, \delta_n^3,
$$
where the constant $C_{\rho}$ depends only on $\rho$. Hence, defining the sequence $u_n := \min \{s, \rho_n \}$, one has that, for every $n \geq 1$ and $\hbar >0$ sufficiently small:
$$
\frac{2 \Vert \varepsilon_\hbar F_{n,\hbar} \Vert_{u_n}}{\delta_n^2} \leq C_{\rho} \, \delta_n \,  \varepsilon_\hbar \leq \frac{1}{2},
$$
where the norm $\Vert \cdot \Vert_{u_n}$ is defined by \eqref{e:compact_norm}.
Using Lemma \ref{flow_estimate}, for every $a \in \mathcal{A}_{s,\rho}(T^*\mathbb{T}^d)$, we have
\begin{equation}
\label{e:egorov_estiamate}
\big \Vert \Psi_{1,\hbar}^{\varepsilon_\hbar F_{n,\hbar}} (a) - a \big \Vert_{u_n - \delta_n} \leq C_\rho \, \delta_n \, \varepsilon_\hbar \Vert a \Vert_{s,\rho}.
\end{equation}
Finally, recalling that $\mathcal{U}_\hbar = \lim_{n \to \infty} U_{n,\hbar} \cdots U_{1,\hbar}$, that every operator $U_{n,\hbar}$ is unitary on $L^2(\mathbb{T}^d)$, and using Lemma \ref{l:analytic_calderon_vaillancourt_torus} and \eqref{e:egorov_estiamate}, we obtain:
\begin{align*}
\Vert \mathcal{U}_\hbar \, \Op_\hbar(a) \, \mathcal{U}_\hbar^* - \Op_\hbar(a) \Vert_{\mathcal{L}(L^2)} & \\[0.2cm]
 & \hspace*{-3cm} \leq C_\rho \sum_{n=1}^\infty \big \Vert \Psi_{1,\hbar}^{\varepsilon_\hbar F_{n,\hbar}}(a) - a \big \Vert_{u_n - \delta_n} \leq C_\rho \, \varepsilon_\hbar \Vert a \Vert_{s,\rho} \sum_{n=1}^\infty \delta_n \leq C_\rho \, \varepsilon_\hbar  \Vert a \Vert_{s,\rho}.
\end{align*}
This shows \eqref{e:first_O}. The proof of \eqref{second_O} is completely analogous but, in this case, using Lemma \ref{flow_estimate} to show that
$$
\big \Vert \Psi_{1,\hbar}^{\varepsilon_\hbar F_{n,\hbar}} (b) - b \big \Vert_{u_n - \delta_n} \leq C_\rho \, \delta_n \, \varepsilon_\hbar \Vert b \Vert_{\mathcal{A}_\rho(\mathbb{T}^d)},
$$
instead of \eqref{e:egorov_estiamate}.
\end{proof}

\begin{lemma}
\label{better_text_functions}
Let $s,\rho > 0$. Assume that $\varepsilon_\hbar = \hbar$ and $V \in \mathcal{A}_{s,\rho}(T^*\mathbb{T}^d)$ satisfies \eqref{e:small_condition3}. Then, for every $a \in \mathcal{C}_c^{\infty}(T^*\mathbb{T}^d)$,
\begin{equation}
\label{e:first_o}
\Vert \mathcal{U}_\hbar \, \Op_\hbar (a) \, \mathcal{U}_\hbar^* - \Op_\hbar( a  ) \Vert_{\mathcal{L}(L^2)} = o(1), \quad  \text{as } \hbar \to 0^+.
\end{equation}
\end{lemma}

\begin{proof}
Let $\varepsilon > 0$ and $a \in \mathcal{C}_c^\infty(T^* \mathbb{T}^d)$.  Assume that, for every $s,\rho > 0$, there exists  $a^\dagger \in \mathcal{A}_{s,\rho}(T^*\mathbb{T}^d)$ such that
\begin{equation}
\label{approximation}
\Vert a - a^\dagger \Vert_{L^\infty(T^*\mathbb{T}^d)} \leq \varepsilon.
\end{equation}
Then, by Lemma \ref{l:O_lemma}, the triangular inequality and \cite[Thm. 13.13]{Zw12}:
\begin{align*}
\Vert \mathcal{U}_\hbar^* \Op_h(a) \, \mathcal{U}_\hbar - \Op_h(a) \Vert_{\mathcal{L}(L^2)} & \\[0.2cm]
& \hspace*{-4.7cm} \leq \Vert \mathcal{U}_\hbar \Op_h(a - a^\dagger )\, \mathcal{U}_\hbar^* \Vert_{\mathcal{L}(L^2)} + \Vert \mathcal{U}_\hbar \Op_h(a^\dagger) \,\mathcal{U}_\hbar^* - \Op_h(a^\dagger ) \Vert_{\mathcal{L}(L^2)} + \Vert \Op_h(a  - a^\dagger ) \Vert_{\mathcal{L}(L^2)} \\[0.2cm]
& \hspace*{-4.7cm} \leq C_d \Vert a - a^\dagger \Vert_{L^\infty(T^* \mathbb{T}^d)} + O_\varepsilon(\hbar),
\end{align*}
and hence
$$
\limsup_{\hbar \to 0^+} \Vert \mathcal{U}_\hbar \Op_\hbar(a) \, \mathcal{U}_\hbar^* - \Op_h(a) \Vert_{\mathcal{L}(L^2)} \leq C_d \, \varepsilon.
$$
Since the choice of $\varepsilon > 0$ was arbitrarily, we conclude that
$$
\lim_{\hbar \to 0^+} \Vert \mathcal{U}_\hbar \Op_h(a) \, \mathcal{U}_\hbar^* - \Op_h(a) \Vert_{\mathcal{L}(L^2)} = 0.
$$
It remains to show \eqref{approximation}. Using the notation \eqref{e:two_fourier} of the Appendix, we write 
$$
a(z) = \frac{1}{(2\pi)^d} \int_{\mathcal{Z}^d} \widehat{a}(w) e^{i z\cdot w} \kappa(dw), \quad z = (x,\xi) \in T^*\mathbb{T}^d.
$$
For $R \geq 1$, we define $a_R \in \mathcal{A}_{s,\rho}(T^*\mathbb{T}^d)$ by
$$
\widehat{a}_R(w) = \widehat{a}(w) e^{-\frac{ \vert w \vert^2}{R}}.
$$
It satisfies
$$
\Vert a_R - a \Vert_{L^\infty(T^*\mathbb{T}^d)} \leq \frac{1}{(2\pi)^d} \int_{\mathcal{Z}^d}\big  \vert \widehat{a}(w) \big \vert \big \vert e^{- \frac{\vert w \vert^2}{R}} - 1 \big \vert \kappa(dw) \to 0, \quad \text{as } R \to \infty.
$$
Then it is sufficient to take $a^\dagger = a_R$ for $R$ sufficiently large.
\end{proof}

\begin{proof}[Proof of Theorem \ref{t:main_result_2}]

By Proposition \ref{p:non_perturbe_operator},
\begin{equation}
\label{e:non_perturbed_measures}
\mathcal{M}(\widehat{L}_{\omega,\hbar}) = \bigcup_{\xi \in \mathcal{L}_\omega^{-1}(1)} \big \{ \mathfrak{h}_{\mathbb{T}^d \times \{ \xi \}} \big \} \cup \{0 \}, \quad \mathcal{N}(\widehat{L}_{\omega,\hbar}) = \left \{ \frac{1}{(2\pi)^d}dx \right \}.  
\end{equation}
On the other hand, Theorem \ref{KAM_theorem_precise} implies that the set of normalized eigenfunctions of $\widehat{Q}_\hbar$ is precisely the orthonormal basis of $L^2(\mathbb{T}^d)$ given by
$$
\{ \Psi_{k,\hbar} =  \mathcal{U}_\hbar^*  e_k \, : \, k \in \mathbb{Z}^d \}.
$$
Using Lemma \ref{better_text_functions}, we obtain that, for every $a \in \mathcal{C}_c^\infty(T^*\mathbb{T}^d)$,
$$
W_{\Psi_{k,\hbar}}^\hbar(a) = W_{e_k}^\hbar(a) + o(1), \quad k \in \mathbb{Z}^d.
$$
Finally, by \eqref{second_O} and since $\mathcal{A}_\rho(\mathbb{T}^d)$ is dense in $\mathcal{C}(\mathbb{T}^d)$, we obtain that, for every $b \in \mathcal{C}(\mathbb{T}^d)$,
$$
\int_{\mathbb{T}^d} b(x) \vert \Psi_{k,\hbar}(x) \vert^2 dx = \int_{\mathbb{T}^d} b(x) \vert e_k(x) \vert^2 dx + o(1).  
$$
Therefore, the proof of the Theorem follows from \eqref{e:non_perturbed_measures}.

\end{proof}

Finally, we complete the proof of Theorem \ref{t:main_result_3}. It follows essentially the same arguments as before, but substituting Lemmas \ref{l:O_lemma} and \ref{better_text_functions} by the following one:

\begin{lemma}
\label{l:last_lemma} Let $\rho > 0$. Let $V \in \mathcal{A}_{W,\rho}(T^*\mathbb{T}^d)$, and assume that $\varepsilon_\hbar \ll \hbar$. Then, for every $a \in \mathcal{C}_c^\infty(T^*\mathbb{T}^d)$,

\begin{equation}
\label{e:first_O2}
\Vert \mathcal{U}_\hbar \, \Op_\hbar(a) \, \mathcal{U}_\hbar^* - \Op_\hbar(a) \Vert_{\mathcal{L}(L^2)} = o(1), \quad \text{as } \hbar \to 0^+,
\end{equation}
and similarly, for every $b \in \mathcal{C}(\mathbb{T}^d)$, 
\begin{equation}
\label{second_O2}
\Vert \mathcal{U}_\hbar \, \Op_\hbar(b) \, \mathcal{U}_\hbar^* - \Op_\hbar(b) \Vert_{\mathcal{L}(L^2)} = o(1), \quad \text{as } \hbar \to 0^+.
\end{equation}

\end{lemma}

\begin{proof}
Recall that, in the case $V \in \mathcal{A}_{W,\rho}(T^*\mathbb{T}^d)$, the first estimate of \eqref{e:induction_hypothesis} remains valid and can be rewriten as
$$
\Vert F_{j,\hbar} \Vert_{W,\rho_j} \leq \frac{\beta \alpha^{\frac{j-1}{2}}}{2}.
$$
Thus, by Corollary \ref{c:same_for_sjostrand}, for every $a \in \mathcal{A}_{W,\rho}(T^*\mathbb{T}^d)$ and $n \geq 1$:
\begin{equation}
\label{e:egorov_estimate2}
\Vert \Psi_{1,\hbar}^{\varepsilon_\hbar F_{n,\hbar}} (a) - a \Vert_{W,\rho_n} \leq O\left( \frac{ \varepsilon_\hbar}{\hbar} \right) \alpha^{ \frac{n-1}{2}} \Vert a \Vert_{W,\rho}.
\end{equation}
Finally, recalling again that $\mathcal{U}_\hbar = \lim_{n \to \infty} U_{n,\hbar} \cdots U_{1,\hbar}$, that every operator $U_{n,\hbar}$ is unitary on $L^2(\mathbb{T}^d)$, and using \eqref{e:calderon_vaillancourt_2} and \eqref{e:egorov_estimate2}, we obtain:
\begin{align*}
\Vert \mathcal{U}_\hbar \, \Op_\hbar(a) \, \mathcal{U}_\hbar^* - \Op_\hbar(a) \Vert_{\mathcal{L}(L^2)} & \\[0.2cm]
 & \hspace*{-3cm} \leq C_\rho \sum_{n=1}^\infty \big \Vert \Psi_{1,\hbar}^{\varepsilon_\hbar F_{n,\hbar}}(a) - a \big \Vert_{W,\rho_n} \leq  O_\rho \left( \frac{ \varepsilon_\hbar}{\hbar} \right) \sum_{n=1}^\infty \alpha^{ \frac{n-1}{2}} \Vert a \Vert_{W,\rho} = O\left( \frac{ \varepsilon_\hbar}{\hbar} \right).
\end{align*}
Therefore, \eqref{e:first_O2} follows by density of $\mathcal{C}_c^\infty(T^*\mathbb{T}^d) \cap \mathcal{A}_{W,\rho}(T^*\mathbb{T}^d)$ in $\mathcal{C}_c^\infty(T^*\mathbb{T}^d)$; and, since $\mathcal{A}_\rho(\mathbb{T}^d) \subset \mathcal{A}_{W,\rho}(T^*\mathbb{T}^d)$, \eqref{second_O2} follows by density of $\mathcal{A}_\rho(\mathbb{T}^d)$ in $\mathcal{C}(\mathbb{T}^d)$.
\end{proof}

\begin{proof}[Proof of Theorem \ref{t:main_result_3}]

Applying Lemma \ref{l:last_lemma}, the proof of Theorem \ref{t:main_result_3} reduces to the proof of Theorem \ref{t:main_result_2}.
\end{proof}

\appendix

\section{Pseudodifferential calculus on the torus}
\label{s:tools_of_analytic_calculus}

We include some basic lemmas about the quantization of the spaces $\mathcal{A}_{s,\rho}(T^*\mathbb{T}^d)$, $\mathcal{A}_s(\R^d)$, $\mathcal{A}_\rho(\mathbb{T}^d)$, $\mathcal{A}_{W,\rho}(T^*\mathbb{T}^d)$ and $S_W(\R^d)$. We fix $s,\rho > 0$ all along this appendix.

 \begin{definition} Let $a : T^*\mathbb{T}^d \to \mathbb{C}$ be a symbol. The semiclassical Weyl quantization $\Op_\hbar(a)$ acting on $\psi \in \mathscr{S}(\mathbb{T}^d)$ is defined by
$$
\Op_\hbar(a) \psi(x) = \sum_{k,j\in\mathbb{Z}^d} \widehat{a}\left( k-j , \frac{ \hbar (k+j)}{2} \right) \widehat{\psi}(k) e^{ij\cdot x},
$$
where $\widehat{a}(k-j,\cdot)$ denotes the $(k-j)^{th}$-Fourier coefficient in the variable $x$.
\end{definition}

\begin{lemma}[Analytic Calderón-Vaillancourt theorem]
\label{l:analytic_calderon_vaillancourt_torus}
For every $a\in \mathcal{A}_{s,\rho}(T^*\mathbb{T}^d)$, the following holds:
\begin{equation}
\label{e:calderon_vaillancourt}
\Vert \Op_\hbar(a) \Vert_{\mathcal{L}(L^2(\mathbb{T}^d))} \leq C_{d,\rho} \Vert a \Vert_{s,\rho}, \quad \hbar \in (0,1].
\end{equation}
\end{lemma}
\begin{proof}
By the usual Calderón-Vaillancourt theorem, see for instance \cite[Prop 3.5]{Pau12}, the following estimate holds:
$$
\Vert \Op_\hbar(a) \Vert_{\mathcal{L}(L^2)} \leq C_d \sum_{\vert \alpha \vert \leq N_d} \Vert \partial_x^\alpha a \Vert_{L^\infty(T^* \mathbb{T}^d)}, \quad \hbar \in (0,1].
$$
Now, using the elementary estimate
\begin{equation}
\label{e:elementary_estimate}
\sup_{t \geq 0} t^m e^{-t\rho} = \left( \frac{m}{e \rho} \right)^m, \quad m > 0,
\end{equation}
we obtain
$$
\Vert \partial_x^\alpha a \Vert_{L^\infty(T^* \mathbb{T}^d)} \leq \frac{1}{(2\pi)^{d/2}} \sum_{k \in \mathbb{Z}^d} \vert k^\alpha \vert \Vert \widehat{a}(k, \cdot) \Vert_{L^\infty(\R^d)} \leq \left( \frac{\vert \alpha \vert}{e \rho} \right)^{\vert \alpha \vert} \Vert a \Vert_{s,\rho} = C_{\alpha, \rho} \Vert a \Vert_{s,\rho}.
$$
\end{proof}

Let $a,b :T^*\mathbb{T}^s \to \mathbb{C}$, the operator given by the composition $\Op_\hbar(a) \Op_\hbar(b)$ is another Weyl pseudodifferential operator with symbol $c$ given by the Moyal product $c = a \sharp_\hbar b$, see for instance \cite[Chp. 7]{Dim_Sjo99}. To write $c$ conveniently, we consider the product space $\mathcal{Z}^d := \mathbb{Z}^d \times \R^d$ and the measure $\kappa$ on $\mathcal{Z}^d$ defined by
$$
\kappa(w) = \mathcal{K}_{\mathbb{Z}^d}(k) \otimes \mathcal{L}_{\R^d}(\eta), \quad w =  (k,\eta) \in \mathcal{Z}^d,
$$
where $\mathcal{L}_{\R^d}$ denotes the Lebesgue measure on $\R^d$, and
$$
\mathcal{K}_{\mathbb{Z}^d}(k) := \sum_{j \in \mathbb{Z}^d} \delta(k-j), \quad k \in \mathbb{Z}^d.
$$
Using this measure, we can write any function $a \in \mathcal{A}_{s,\rho}(T^*\mathbb{T}^d)$ as
\begin{equation}
\label{e:two_fourier}
a(z) = \frac{1}{(2\pi)^d} \int_{\mathcal{Z}^d} \mathscr{F} a(w) e^{iz \cdot w} \kappa(dw),
\end{equation}
where $z = (x,\xi) \in T^*\mathbb{T}^d$, and $\mathscr{F}$ denotes the Fourier transform in $T^*\mathbb{T}^d$:
$$
\mathscr{F} a(w) = \frac{1}{(2\pi)^d} \int_{T^*\mathbb{T}^d}  a(z) e^{-i w\cdot z} dz.
$$
With these conventions, the Moyal product $c = a \sharp_\hbar b$ can be written by the following  integral formula:
\begin{align}
\label{e:moyal_product}
a \sharp_\hbar b(z) = \frac{1}{(2\pi)^{2d}}\int_{\mathcal{Z}^{d} \times \mathcal{Z}^d} \big( \mathscr{F} a \big) (w') \big( \mathscr{F} b \big)(w-w') e^{\frac{i \hbar}{2} \{ w', w-w'\}} e^{i z \cdot w} \kappa(dw') \kappa(dw),
\end{align}
where $\{\cdot , \cdot  \}$ stands for the standard symplectic product in $\mathcal{Z}^d \times \mathcal{Z}^d$:
$$
\{w,w' \} = k \cdot \eta' - k' \cdot \eta, \quad w = (k,\eta), \quad w' = (k',\eta').
$$
Alternatively, we can deduce from \eqref{e:moyal_product} the following formula:
\begin{equation}
\label{e:moyal_product_2}
a \sharp_\hbar b(x,\xi) = \frac{1}{(2\pi)^d} \sum_{k,k'\in \mathbb{Z}^d} \widehat{a} \left(k', \xi + \frac{\hbar(k-k')}{2} \right) \widehat{b} \left( k-k', \xi - \frac{ \hbar k'}{2} \right) e^{i k \cdot x},
\end{equation}
which also holds for symbols $a,b \in \mathcal{A}_{W,\rho}(T^*\mathbb{T}^d)$.

\begin{lemma}
\label{l:moyal_sjostrand_bound}
Let $\rho > 0$ and $a,b \in \mathcal{A}_{W,\rho}(T^*\mathbb{T}^d)$. Then $a\sharp_\hbar b \in \mathcal{A}_{W,\rho}(T^*\mathbb{T}^d)$ and
$$
\Vert a \sharp_\hbar b \Vert_{W,\rho} \leq \Vert a \Vert_{W,\rho} \Vert b \Vert_{W,\rho},
$$
provided that the constant $C_d$ in \eqref{e:large_constant} is sufficiently large.
\end{lemma}

The proof of Lemma \ref{l:moyal_sjostrand_bound} can be obtained following the same ideas of \cite[Thm. 2.1]{Sjo94}. We omit the details for the sake of shortness. Moreover, one can show $L^2$ boundness for pseudodifferential operators with symbols in $\mathcal{A}_{W,\rho}(T^*\mathbb{T}^d)$. Precisely, the following  Calderón-Vailancourt-type result holds (see \cite[Sect. 3]{Sjo94}): for every $a \in \mathcal{A}_{W,\rho}(T^*\mathbb{T}^d)$:
\begin{equation}
\label{e:calderon_vaillancourt_2}
\Vert \Op_\hbar(a) \Vert_{\mathcal{L}(L^2)} \leq C_{d,\rho} \Vert a \Vert_{W,\rho}, \quad \hbar \in (0,1].
\end{equation}

Coming back to the symbolic calculus for Weyl pseudodifferential operators, we will employ the notation 
\begin{equation}
\label{e:commutator_symbol}
[a,b]_\hbar := a \sharp_\hbar b - b \sharp_\hbar a,
\end{equation}
for the Moyal commutator. Hence $[\Op_\hbar(a), \Op_\hbar(b)] = \Op_\hbar([a,b]_\hbar)$. Moreover, given two symbols $a, F \in \mathcal{A}_{W,\rho}(T^*\mathbb{T}^d)$, we have the following formula for the conjugation of $\Op_\hbar(a)$ by $e^{i\frac{t}{\hbar} \Op_\hbar(F)}$:
$$
e^{i\frac{t}{\hbar} \Op_\hbar(F)} \Op_\hbar(a) e^{-i\frac{t}{\hbar} \Op_\hbar(F)} = \Op_\hbar \big( \Psi_{t,\hbar}^F(a) \big), \quad t \in [0,1],
$$
where the symbol $\Psi_{t,\hbar}^F(a)$  is given  by
\begin{equation}
\label{e:conjugation_symbol}
 \Psi_{t,\hbar}^F(a) := \sum_{j = 0}^\infty \frac{1}{j!} \left( \frac{it}{\hbar} \right) \Ad_{F}^{\sharp_\hbar,j}(a), \quad t \in [0,1],
\end{equation}
and, as usual in the terminology of Lie algebras,
$$
\Ad_{F}^{\sharp_\hbar,j}(a) = [F, \Ad_{F}^{\sharp_\hbar,j-1}(a)]_\hbar, \quad \Ad_{F}^{\sharp_\hbar,0}(a) = a.
$$

\begin{lemma}
\label{l:bad_flow_estimate}
Assume that $a, F \in \mathcal{A}_{s,\rho}(T^*\mathbb{T}^d)$. Let $\varepsilon_\hbar \leq \hbar$ such that
$$ 
 \beta =  2  \Vert F \Vert_{s,\rho} \leq 1/2,
$$ 
then 
$$
\Vert  \Psi_{t,\hbar}^{\varepsilon_\hbar F}(a) - a \Vert_{s,\rho} \leq \beta \Vert a \Vert_{s,\rho}, \quad  \vert t \vert \leq 1.
$$
\end{lemma}

\begin{proof}
Using the expression
$$
[a,b]_\hbar(z) = 2i \int_{\mathcal{Z}^{2d}} \mathscr{F}a(w') \mathscr{F}b(w-w') \sin \left( \frac{\hbar}{2} \{ w', w-w' \} \right) \frac{e^{iw\cdot z}}{(2\pi)^{2d}} \kappa(dw') \, \kappa(dw),
$$
for the Moyal commutator, we obtain that
\begin{equation}
\label{trivial_bound}
\Vert [F,a]_\hbar \Vert_{s,\rho} \leq 2 \Vert F \Vert_{s,\rho} \Vert a \Vert_{s,\rho}.
\end{equation}
Applying this succesively, we conclude that
$$
\Vert \Psi_{t,\hbar}^{\varepsilon_\hbar F}(a) - a \Vert_{s,\rho} \leq \sum_{j=1}^\infty \frac{1}{j!} \left( \frac{t}{\hbar} \right)^j \Vert \Ad_{\varepsilon_\hbar F}^{\sharp_\hbar,j} (a) \Vert_{s,\rho} \leq \sum_{j=1}^\infty \frac{2^j \Vert F \Vert_s^j  \Vert a \Vert_s}{j!}   \leq \beta \Vert a \Vert_{s,\rho}.
$$
\end{proof}

\begin{corol}
\label{c:same_for_sjostrand}
The same holds replacing the space $\mathcal{A}_{s,\rho}(T^*\mathbb{T}^d)$ by $\mathcal{A}_{W,\rho}(T^*\mathbb{T}^d)$.
\end{corol}

\begin{proof}
Note that, by Lemma \ref{l:moyal_sjostrand_bound},
\begin{equation}
\label{trivial_bound2}
\Vert [F,a]_\hbar \Vert_{W,\rho} \leq 2 \Vert F \Vert_{W,\rho} \Vert a \Vert_{W,\rho}.
\end{equation}

\end{proof}

Finally, we will improve the above estimates for  our analytic spaces, allowing some loss of analyticity. Let $0 < u \leq \min \{s,\rho \}$, we denote
\begin{equation}
\label{e:compact_norm}
\Vert a \Vert_{u} := \frac{1}{(2\pi)^d} \int_{\mathcal{Z}^d} \vert \mathscr{F}a(w) \vert e^{\vert w \vert u} \kappa(dw) = \Vert a \Vert_{u,u} \leq \Vert a \Vert_{s,\rho}.
\end{equation}
\begin{lemma}
\label{l:commutator_loss}
Let $a,b \in \mathcal{A}_{s,\rho}(T^*\mathbb{T}^d)$. Then, for $0 < \sigma_1 + \sigma_2 < u := \min \{ s, \rho \}$:
\begin{equation}
\label{e:first_commutator}
\big \Vert [a,b]_\hbar \big \Vert_{u-\sigma_1-\sigma_2} \leq \frac{2\hbar}{e^2 \sigma_1(\sigma_1 + \sigma_2)}  \Vert a \Vert_{u} \Vert b \Vert_{u-\sigma_2}. 
\end{equation}
Moreover, if  $c \in \mathcal{A}_\rho(\mathbb{T}^d)$, then:
\begin{equation}
\label{e:second_commutator}
\big \Vert [a,c]_\hbar \big \Vert_{u-\sigma_1-\sigma_2} \leq \frac{\hbar}{e^2 \sigma_1(\sigma_1 + \sigma_2)}  \Vert a \Vert_{u} \Vert c \Vert_{\mathcal{A}_{\rho-\sigma_2}(\mathbb{T}^d)}.
\end{equation}
\end{lemma}

\begin{proof}
By \eqref{e:moyal_product}, we have
$$
[a,b]_\hbar(z) = 2i \int_{\mathcal{Z}^{2d}} \mathscr{F}a(w') \mathscr{F}b(w-w') \sin \left( \frac{\hbar}{2} \{ w', w-w' \} \right) \frac{e^{iw\cdot z}}{(2\pi)^{2d}} \kappa(dw') \, \kappa(dw).
$$
Then, using that
\begin{equation}
\label{e:bound_symplecti_product}
\vert \{w', w-w' \} \vert \leq 2 \vert w' \vert \vert w- w' \vert,
\end{equation}
we obtain:
\begin{align*}
\big \Vert [a,b]_\hbar \big \Vert_{u-\sigma_1-\sigma_2} & \\[0.2cm]
 & \hspace*{-2cm} \leq 2\hbar \int_{\mathcal{Z}^{2d}} \vert \mathscr{F}a(w') \vert \vert w' \vert \vert \mathscr{F}b(w-w') \vert \vert w - w' \vert e^{(u-\sigma_1 - \sigma_2)(\vert w-w' \vert + \vert w' \vert)} \kappa(dw') \kappa(dw) \\[0.2cm] 
 & \hspace*{-2cm} \leq 2\hbar \Big( \sup_{r\geq 0} r e^{-\sigma_1r} \Big) \Big( \sup_{r \geq 0} r e^{-(\sigma_1 + \sigma_2)r} \Big) \Vert a \Vert_{u} \Vert b \Vert_{u - \sigma_2} \\[0.2cm]
 & \hspace*{-2cm} \leq \frac{2\hbar}{e^2 \sigma_1 (\sigma_1 + \sigma_2)} \Vert a \Vert_{u} \Vert b \Vert_{u-\sigma_2}.
\end{align*}
To prove \eqref{e:second_commutator}, observe that, in view of \eqref{e:moyal_product} and \eqref{e:moyal_product_2},
$$
[a,c]_\hbar(z) = 2i \sum_{k \in \mathbb{Z}^d} \int_{\mathcal{Z}^{d}} \mathscr{F}a(w') \widehat{c}(k-k') \sin \left( \frac{\hbar}{2} (k-k')\cdot \eta' \right) e^{i(k\cdot x + \eta' \cdot \xi) } \frac{\kappa(dw') }{(2\pi)^{3d/2}} .
$$
Then \eqref{e:second_commutator} follows by the  the same argument as before but with the estimate
$$
\vert (k-k') \cdot \eta' \vert \leq \vert w' \vert \vert k - k' \vert,
$$
instead of \eqref{e:bound_symplecti_product}.
\end{proof}

\begin{lemma}
\label{flow_estimate}
Assume $a, F \in \mathcal{A}_{s,\rho}(T^*\mathbb{T}^d)$ and $b \in \mathcal{A}_\rho(\mathbb{T}^d)$. Let $0 < \sigma < u := \min \{s,\rho \}$. Assume that
$$ 
 \beta =  \frac{2 \Vert F \Vert_{u}}{\sigma^2} \leq 1/2.
$$ 
Then 
\begin{equation}
\label{e:double_bound}
\Vert \Psi_{t,\hbar}^{F}(a) - a \Vert_{u-\sigma} \leq \beta \Vert a \Vert_{u}, \quad \Vert \Psi_{t,\hbar}^{F}(b) - b  \Vert_{u-\sigma} \leq \beta \Vert b \Vert_{\mathcal{A}_\rho(\mathbb{T}^d)}, \quad  \vert t \vert \leq 1.
\end{equation}
\end{lemma}

\begin{proof}
By estimate \eqref{e:first_commutator}, for every $j \geq 1$,
\begin{align*}
\big \Vert \Ad^{\sharp_\hbar,j}_{F}(a) \big \Vert_{u-\sigma} & = \big \Vert [F, \Ad^{\sharp_\hbar,j-1}_{F}(a)]_\hbar \big \Vert_{u-\sigma} \\[0.2cm]
 & \leq  \frac{2\hbar j}{e^2 \sigma^2} \Vert F \Vert_{u} \big \Vert \Ad^{\sharp_\hbar,j-1}_{F}(a) \Vert_{u- \sigma (j-1)/j} \\[0.2cm]
 & \leq \frac{4 \hbar^2 j^3}{e^4 \sigma^4 (j-1)} \Vert F \Vert_{u}^2 \big \Vert \Ad^{\sharp_\hbar,j-2}_{F}(a) \big \Vert_{u - \sigma(j-2)/j} \\[0.2cm]
 & \leq \cdots \leq  \left( \frac{2 \hbar}{e^2 \sigma^2} \right)^j \frac{j^{2j}}{j!} \Vert F \Vert_{u}^j \Vert a \Vert_{u}.
\end{align*}
Using Stirling's formula: $j^j/e^{j-1} j! \leq 1$ for $j \geq 1$, we conclude that
$$
\Vert \Psi_{t,\hbar}^{F}( a) - a \Vert_{u-\sigma} \leq   \frac{\Vert a \Vert_{u}}{e^2}\sum_{j =1}^\infty\beta^j  \leq \beta \Vert a \Vert_{u}.
$$
The proof of the second inequality in \eqref{e:double_bound} follows by the same argument, but using \eqref{e:second_commutator} instead of \eqref{e:first_commutator} to obtain
$$
\Vert [F, b ] \Vert_{u- \sigma/j} \leq \frac{ \hbar j^2}{e^2 \sigma^2} \Vert F \Vert_{u} \Vert b \Vert_{\mathcal{A}_\rho(\mathbb{T}^d)}.
$$ 
\end{proof}

\bibliographystyle{plain}
\bibliography{Referencias}
\medskip

\begin{flushleft}
\small{\textit{Instituto de Ciencias Matemáticas (ICMAT-UAM-UC3M-UCM)}} \\
\small{\textit{C/ Nicolás Cabrera, nº 13-15 Campus de Cantoblanco, UAM,}} \\
\small{\textit{28049 Madrid, Spain.}} \\
\small{\textit{victor.arnaiz@icmat.es}}
\end{flushleft}

\end{document}